\documentclass{article}

\usepackage{microtype}
\usepackage{graphicx}
\usepackage{subfigure}
\usepackage{booktabs} % for professional tables

\usepackage{hyperref}

\usepackage[accepted]{icml2024}

\usepackage{amsmath}
\usepackage{amssymb}
\usepackage{mathtools}
\usepackage{amsthm}
\usepackage{thm-restate}

\usepackage[capitalize,noabbrev]{cleveref}

\theoremstyle{plain}
\newtheorem{theorem}{Theorem}[section]
\newtheorem{proposition}[theorem]{Proposition}
\newtheorem{lemma}[theorem]{Lemma}

\theoremstyle{definition}
\newtheorem{definition}[theorem]{Definition}

\theoremstyle{remark}

\usepackage[textsize=tiny]{todonotes}

\usepackage{bm, bbm}
\usepackage{lipsum, color}
\usepackage{nicefrac}
\usepackage{parskip}
\usepackage{color-edits}

\addauthor{joon}{blue}
\addauthor{kk}{magenta}

\usepackage[many]{tcolorbox}
\newtcolorbox{framefloat}[1][!tb]{arc=0pt,outer arc=0pt,boxrule=0.4pt,colframe=black,colback=white,float=#1}

\def \A {\mathcal{A}}

\def \D {\mathcal{D}}

\def \M {\mathcal{M}}

\def \P {\mathcal{P}}

\def \T {\mathcal{T}}
\def \U {\mathcal{U}}

\def \unit {[0,1]}
\def \unito {[0,1)}

\def \munit {[-1,1]}
\def \nat {\mathbb{N}}
\def \rr {\mathbb{R}}

\DeclareMathOperator*{\expec}{\mathbb{E}}
\DeclareMathOperator*{\pr}{\text{Pr}}
\DeclareMathOperator*{\argmax}{arg\,max}
\DeclareMathOperator*{\argmin}{arg\,min}

\newcommand{\cbra}[1]{\left\{ #1 \right\}}
\newcommand{\sbra}[1]{\left[ #1 \right]}
\newcommand{\rbra}[1]{\left( #1 \right)}

\newcommand{\ind}{\mathbbm{1}}
\newcommand{\bunderbrace}[2]{%
  \begin{array}[t]{@{}c@{}}
  \underbrace{#1}\\
  #2
  \end{array}
}
\newcommand\numberthis{\addtocounter{equation}{1}\tag{\theequation}}

\allowdisplaybreaks
\newcounter{protocol}
\def \alg {\textup{Alg}}
\def \opt {\textup{Opt}}
\def \reg {\textup{Reg}}
\def \Mhedge {\M^\textup{Hedge}}
\def \Mwdp {\M^\textup{wDP}}
\def \qwdp {q^{\eta\textup{--wdp}}}
\def \pisr {\pi^\textup{sr}}
\def \picom {\pi^\textup{com}}
\def \pipl {\pi^\textup{pl}}
\def \pivcg {\pi^\textup{vcg}}
\def \pipp {\pi^\textup{pp}}
\def \Pipl {\Pi^\textup{PL}}
\def \Pivcg {\Pi^\textup{VCG}}

\def \Pimaxmin {\Pi^\textup{MMF}}

\def \qhedge {q^\textup{hedge}}
\def \DSIC {\textup{DSIC}}
\def \NIC {\textup{NIC}}

\def \NICOM {\textup{NICOM}}

\ifx\example\undefined
\newtheorem*{example}{Example}
\fi

\icmltitlerunning{Nash Incentive-compatible Online Mechanism Learning}

\begin{document}

\twocolumn[
%% Title candidates
\icmltitle{Nash Incentive-compatible Online Mechanism Learning \\ via Weakly Differentially Private Online Learning}

\icmlsetsymbol{equal}{*}

\begin{icmlauthorlist}
\icmlauthor{Joon Suk Huh}{sch}
\icmlauthor{Kirthevasan Kandasamy}{sch}
\end{icmlauthorlist}

\icmlaffiliation{sch}{Department of Computer Science, University of Wisconsin--Madison, Madison, WI, USA}

\icmlcorrespondingauthor{Joon Suk Huh}{joon@cs.wisc.edu}

% You may provide any keywords that you
% find helpful for describing your paper; these are used to populate
% the "keywords" metadata in the PDF but will not be shown in the document
\icmlkeywords{Machine Learning, Mechanism Design, Auction, Algorithmic Game Theory, Differential Privacy, Online Learning, Learning Theory, ICML}

\vskip 0.3in
]

\printAffiliationsAndNotice{}  % leave blank if no need to mention equal contribution
%\printAffiliationsAndNotice{\icmlEqualContribution} % otherwise use the standard text.

%%%%%%%% BODY %%%%%%%%
\begin{abstract}
We study a multi-round mechanism design problem, where we interact with a set of agents over a sequence of rounds.
We wish to design an incentive-compatible (IC) online learning scheme to maximize an application-specific objective within a given class of mechanisms, without prior knowledge of the agents' type distributions.
Even if each mechanism in this class is IC in a single round, if an algorithm naively chooses from this class on each round, the entire learning process may not be IC against non-myopic buyers who appear over multiple rounds.
On each round, our method randomly chooses between the recommendation of a weakly differentially private online learning algorithm (e.g., Hedge), and a commitment mechanism which penalizes non-truthful behavior.
Our method is IC and achieves $O(T^{\frac{1+h}{2}})$ regret for the application-specific objective in an adversarial setting, where  $h$ quantifies the long-sightedness of the agents.
When compared to prior work, our approach is conceptually simpler,
it applies to general mechanism design problems (beyond auctions), and its regret scales gracefully with the size of the mechanism class.
\end{abstract}
\section{Introduction}\label{sec:intro} 
In mechanism design, a social planner designs a protocol,
where a set of agents report their \emph{type} to the mechanism, and the mechanism chooses an outcome based on the reported types. It is often desirable, if not necessary, that a mechanism be \textit{incentive-compatible} (IC), i.e. it is in the best interest for each agent to report their types truthfully. When designing the mechanism, we are interested in optimizing some application-specific objectives.
This typically requires knowledge of the distribution on agents' types. As examples, consider a seller who wishes to design either a reserve price auction or a posted price mechanism to sell an item. In both cases,
to maximize revenue (objective), the seller requires knowledge of the distribution of buyers' (agents) valuations (types) for the item.

However, in practice, this type distribution is often unknown \textit{a priori}. Consequently, learning an optimal mechanism from data has been one of the central themes in the intersection of machine learning and game theory. Several prior works have studied learning optimal mechanisms from an i.i.d. dataset of types in offline settings~\citep{cole2014sample,daskalakis2014complexity,balcan2016sample,balcan2018general}. In this work, our focus is on online settings, where there is a stream of types and objectives, and the goal is to design \emph{no-regret} learning procedures which need to trade-off between exploring candidate mechanisms in a mechanism class, while simultaneously exploiting current knowledge to maximize the desired objective~\citep{kleinberg2003value,blum2005near}.

This problem of online mechanism learning has a distinctive challenge: the data sources are themselves strategic agents endowed with their own \emph{long-term} utilities when interacting with the mechanism over multiple rounds. These agents may misreport their type for their own long-term benefit.
This is detrimental to learners who will be unable to estimate the true frequency of types and optimize the given objective.
In particular, even if a multi-round mechanism is IC in a single round, it does not necessarily translate to IC over the \emph{entire} learning process. Agents may find it beneficial to misreport their type in the early rounds while incurring some short-term loss, so as to manipulate later outcomes in their favor, resulting in long-term gains.

To illustrate this, consider designing a revenue-optimal posted price mechanism (PPM)\footnote{%
This description is different, yet equivalent to another common interpretation of the PPM, where the seller announces the price $p$ and buyers buy if their value is larger than $p$.
}, where a seller has an infinite supply of a single item. The seller chooses a price $p$ without revealing it to the buyers. The buyers report their value, and the seller sells the item to those whose value exceeds $p$ and charges them $p$. This mechanism is IC on any single round. However, if the seller naively repeats this over a sequence of rounds, and uses reported values from previous rounds to set the current price, buyers may bid lower in the early rounds. While this deceitful bidding strategy may risk not getting the item early on, it would misguide a learning algorithm to think the current price is too high, and lower future prices, benefiting the buyer.
Consequently, the seller will not be able to optimize her long-term revenue.

% Motivated by the above challenges, we investigate the following question in this work: 
% \begin{center}
%     \textit{Is there an IC online learning scheme for every\\ IC (single-round) mechanism class?}
% \end{center}
Recent works have recognized this challenge and developed learning algorithms for auctions with various strategy-robustness and IC guarantees~\citep{amin2013learning,mohri2014optimal,mohri2015revenue,drutsa2017horizon,drutsa2018weakly,liu2018learning,drutsa2020reserve,abernethy2019learning,deng2020robust,mirrokni2020non}.
However, all methods above only apply to auctions, and not general mechanism design problems.
Moreover,  in most work, there is only a single parameter to be learned in the mechanism (e.g. reserve price in an auction, selling price in a PPM), and there is no natural extension when there are multiple learnable parameters (e.g. VCG auctions with several reserve prices~\citep{hartline2009simple,roughgarden2019minimizing}).

Our approach offers the following advantages when compared to prior work.
First, it applies to general mechanism design problems with application-specific objectives, including mechanisms without money~\citep{procaccia2013approximate,nissim2012approximately}.
Second, it can be naturally applied to settings with multiple learnable parameters. In particular, by leveraging a weak notion of differential privacy, we are able to avoid the curse of dimensionality in large parameter spaces.
Third, our method is conceptually much simpler than many of the existing work, as it is based on the well-known Hedge algorithm~\citep{hazan2016introduction}.
%%%%%%%%%%%%%%%%%%%%%%%%%%%%%%%%%%%%%%%%%%%%%%%%%%%%%%%%%%%%%%%%%%
\subsection{Summary of Contributions and Main Results}\label{subsec:contrib}

Our main contributions in this work are as follows:
\begin{enumerate}
    \item In~\S\ref{sec:icoml}, we present a framework for IC online learning of mechanisms with full-information feedback.
    In our framework, an \emph{online mechanism} chooses a single-round mechanism on each round; an environment realizes agents' types and the social planner's objective. Agents report their types, and then the single-round mechanism chooses an outcome (see Protocol \ref{protocol:1}).
    Our framework applies to general mechanism design problems where a social planner may have different objectives on each round, and agents may have different types on each round.
    
    \item In~\S\ref{subsec:design-nicom}, we show that, under certain conditions on the agents' long-sightedness (Definition \ref{def:long-sightedness}), any \emph{weakly} differentially private (DP) sequence (Definition \ref{def:weak-dp-seq}) of distributions over mechanisms can be used as a primitive for building an IC online learning mechanism. Our notion of a weak DP sequence is a considerably weaker technical requirement than the usual notion of DP for sequential algorithms used to build IC online mechanisms in prior work \cite{jain2012differentially,abernethy2019learning}.
    
    \item In~\S\ref{subsec:learning-via-hedge}, we use the above result to construct an IC online learning mechanism using the well-known Hedge (a.k.a. exponential weights) algorithm \cite{hazan2016introduction}. As we show, Hedge provides a weak DP sequence. 
    Our learning algorithm is \emph{(i)} Nash incentive-compatible (NIC, Definition~\ref{def:nic-online}) and \emph{(ii)} achieves $O\big((\log |\Pi|+\beta(\picom)^{-1})\cdot T^{(1+h)/2}\big)$ regret compared to the best mechanism in any NIC mechanism class $\Pi$. Here $h\in\unito$ is an exponent that characterizes the long-sightedness of agents (Definition \ref{def:long-sightedness} and Theorem \ref{thm:hedge-mech}), and $\beta(\picom)$ is the penalty gap of the commitment mechanism (Definition~\ref{def:penalty-gap}).
    
    \item In~\S\ref{sec:examples} and Appendix \ref{app:resource-alloc}, we apply our framework to develop no-regret NIC online learning mechanisms for various mechanism design problems. They are \emph{(i)} the $k$-facility location problem \cite{lu2010asymptotically,fotakis2013winner} where the goal is to minimize the average distance between citizens and their closest facilities, \emph{(ii)} Social welfare maximization with \textit{ex-post} externality via the VCG mechanism with bidder-specific reserve prices~\citep{hartline2009simple,roughgarden2019minimizing}, and \emph{(iii)} a resource allocation problem~\citep{demers1989analysis} where we wish to maximize the number of jobs completed.
\end{enumerate}

\textit{Key insight.} While prior works have used DP to achieve IC for auctions in sequential settings, the standard notion of DP for sequential algorithms~\citep{jain2012differentially} is stronger than necessary. 
%For instance, using standard DP online learning algorithms will result in regret with linear dependence in $|\Pi|$  (e.g. see Table 1 in \citet{agarwal2017price}) in the adversarial setting.
Our key insight is in recognizing that a weaker notion of DP for sequential algorithms is sufficient, and moreover that the vanilla Hedge algorithm satisfies this property. This helps us reduce the dependence to logarithmic in $|\Pi|$. Conveniently, this approach also provides a straightforward extension to general mechanism design problems, beyond auctions.

\textit{Algorithm overview.} On each round, our construction randomly chooses a lottery of two options. In the first, we follow a mechanism in $\Pi$ recommended by Hedge, and in the other, we implement a problem-specific \emph{commitment mechanism} that penalizes any non-truthful behavior~\cite{nissim2012approximately}. Then, the agents report their types after seeing the lottery of mechanisms, which includes the probability of each option that is set based on problem-specific factors.

\textit{Proof overview.}
On each round, the recommendation from Hedge can be viewed as a sequence of randomized maps from past observations to a single-round mechanism. We show that \textit{each} randomized map in the sequence being $\eta$-DP is enough to guarantee that truth-telling is a Nash equilibrium in the online setting (Theorem \ref{thm:ic-via-dp}), for some choice of $\eta\geq 0$.
Specifically, while agents can gain by lying in Hedge, the weak DP property ensures that this gain is bounded. Then, the commitment mechanism cancels out this gain by sufficiently penalizing untruthful reporting.
While similar intuitions have been used in prior auction work,
the key challenge of applying this idea with our notion of weak DP is that the agents' potential gains are functions of the full sequence of outcomes up to the final round $T$. Hence, using sequential composition \cite{dwork2014algorithmic}, we can only bound the agents' potential gains by an $\eta T$ factor of the maximum utility, which grows linearly with $T$.
%the number of rounds. 

To address this, we rely on a backward-induction argument. First, we show that an agent has no incentive to deviate from truth-telling in the last round. Then, we recursively apply this argument to demonstrate that truth-telling is the best strategy in all rounds. The key observation is that our weak DP notion is sufficient for this proof technique to work.
% To circumvent this, we carefully decompose the agents' utilities in a round-by-round manner, establishing that our weak DP notion is sufficient to guarantee a small upper bound on the agents' gains from lying.

%%%%%%%%%%%%%%%%%%%%%%%%%%%%%%%%%%%%%%%%%%%%%%%%%%%%%%%%%%%%%%%%%%
\subsection{Discussion of Most Relevant Works}\label{subsec:relevant-work}
Due to space constraints, we discuss only the most relevant works here and give a broader overview in Appendix~\ref{app:related-work}. Section~\ref{subsec:discussion-context-reg} reviews previous regret bounds for settings similar to ours~\citep{abernethy2019learning,liu2018learning,amin2013learning}.

One of the closest lines of research to ours is dynamic mechanism design~\citep{baron1984regulation,thomas1990income,courty2000sequential,bergemann2010dynamic,ashlagi2016sequential,mirrokni2016dynamic,mirrokni2020non,bergemann2019dynamic,deng2019non}. In dynamic mechanism design, the seller possesses precise knowledge of the buyer's value distribution across all rounds. Leveraging this information, the seller designs a revenue-maximizing dynamic mechanism by adjusting the mechanism's rule based on the buyers' past bids. Recently, dynamic mechanism design has been extended to the setting where the exact distributional knowledge requirement is relaxed to approximate ones~\citep{golrezaei2019dynamic,deng2019robust,deng2020robust}. However, even with these extensions, dynamic mechanism design requires a strong assumption of the existence of an underlying distribution on buyers' values, and often necessitates all buyers to share their beliefs. Our approach does not require any of this: We allow the types of agents to be arbitrary sequences of time. Moreover, our Nash incentive-compatibility guarantee does not require agents to share their exact beliefs other than everyone reporting their types truthfully. Furthermore, to the best of our knowledge, dynamic mechanism design has primarily focused on mechanisms involving monetary transactions, whereas our framework operates independently of any payment or quasilinear utility assumptions.
 
Besides the dynamic mechanism design literature, the closest work to ours is \citet{abernethy2019learning}, who delve into an online Bayesian optimal auction design problem. They use a strongly DP online learning algorithm to design a version of the Myerson auction that sells a limited supply of identical items to multiple buyers. They show that this approach is $\eta$-bid-approximate Bayesian incentive-compatible, meaning that every agent submitting bids within $\eta$ proximity to their true values is a Nash equilibrium with respect to each agent's Bayesian belief. Our IC guarantee, Nash incentive compatibility, is stronger as it does not depend on any prior type distribution and holds for all true agent types (i.e. in an \emph{ex-post} sense). Moreover, our IC guarantee is valid even when all previous information is \textit{public} to all agents, while previous works \citep{liu2018learning,abernethy2019learning} rely on the private information assumption where agents cannot observe other agents' previous bids.

To realize an incentive-compatible online mechanism, our method uses a lottery between an online learning algorithm and a \emph{commitment mechanism} which penalizes agents for non-truthful behavior.
The technique of combining a differentially private algorithm with a commitment mechanism was originally introduced by~\citet{nissim2012approximately} to design approximately optimal single-round mechanisms. We extend their approach to the online setting without losing efficiency in the sense that our regret bound is logarithmic in $|\Pi|$. 

%This is in stark contrast with a naive generalization of the idea of \citet{nissim2012approximately} where one simply replaces the exponential mechanism with a (strongly) differentially private online expert algorithm such as one given in \citet{agarwal2017price}. In this case, one would obtain an algorithm regret that is linear to $|\Pi|$.

%ideas have been used by~\citet{nissim2012approximately} to design a general single-round mechanism. A straightforward generalization of their techniques to the online setting would be using (strongly) differentially private online expert algorithms such as one given in \citet{agarwal2017price}. However, as we discussed earlier, such a naive construction led to the regret 
%\kkcomment{Discuss further. Challenges in the online setting, straightforward extension is strong-DP which is bad?}

\section{Background}\label{sec:prelim}

In this section, we provide a brief overview of mechanism design and differential privacy.
\paragraph{Notation and assumptions.}
For $d\in\nat$, let $[d]:=\cbra{1,\cdots,d}$. Let $\ind\!\sbra{\,\cdot\,}$ be the indicator function. For a set $S$, let $\Delta(S)$ denote the set of probability distributions over $S$. We use $t\in[T]$ ($i\in[n]$) to denote rounds (agents). For any set of round indexed variable $\{x_t\}$, let $x_{<t}:=(x_\tau)_{\tau\in[t-1]}$ and $x_{\leq t}:=(x_\tau)_{\tau\in[t]}$. Throughout this paper, we assume types of agents are discrete, that is, we set $\Theta_i$ introduced in the next subsection as a countable set.

%%%%%%%%%%%%%%%%%%%%%%%%%%%%%%%%%%%%%%%%%%%%%%%%%%%%%%%%%%%%%%%%%%
\subsection{Mechanism Design}\label{subsec:mech-design}
There are $n$ agents, denoted $[n]$. Agent $i$ has type $\theta_i\in\Theta_i$, where $\Theta_i$ is the set of available types.
Each agent reports their type to a mechanism, which may not necessarily be the same as their true type $\theta_i$. We use $b_i\in\Theta_i$ to denote the type reported by agent $i$. Based on the reported types, the mechanism, possibly randomly, chooses an outcome $s\in S$, where $S$ is the set of available outcomes.
Then, each agent $i$ experiences utility $u_i(\theta_i,s)\in[-1,1]$.
The social planner's satisfaction is given by an (application-specific) \emph{objective} $G(\theta,s)\in[-1,1]$, where $\theta:=(\theta_i)_{i\in[n]}$ is all agent types.

A mechanism $\pi$ is a map from reported types to a distribution over outcomes in $S$, i.e. $\pi:\prod_{i\in[n]}\Theta_i  \rightarrow\Delta(S)$, where $\pi(b)$ is the distribution induced on the set of outcomes when the agents report $b:=(b_i)_{i\in[n]}$. An outcome $s$ is then sampled from $\pi(b)$.
The goal of mechanism design is to design $\pi$ which has large expected objective $\expec_{s\sim\pi(\theta)}G(\theta,s)$, when the agents are reporting truthfully, i.e. $b_i = \theta_i$ for all $i\in[n]$.

However, as agents may misreport their types to maximize their utility, a mechanism should also be IC, i.e. incentivizes agents to report their types truthfully. To state this formally, with slight abuse of notation, let
\begin{align}
    u_i\big(b;\pi,\theta_i\big):=\expec_{s\sim \pi(b_i,b_{-i})}\left[u_i(\theta_i,s)\right],
    \label{eq:expectedutilitysingle}
\end{align}
denote the expected utility of agent $i$ of type $\theta_i$, given reported types $b$ under mechanism $\pi$.
Here, $u_i$ on the RHS is the utility of the agent for outcome $s$ as defined above.

One common notion of IC is the dominant-strategy incentive-compatibility (\DSIC). We say that $\pi$ is \DSIC\ if truth-telling is better for all agents regardless of the types reported by the other agents. That is, for all $i\in[n]$, $\theta_i,b_i\in\Theta_i$ and $b_{-i}\in\Theta_{-i}:=\prod_{j\in[n]\setminus\{i\}}\Theta_j$,
we have
\begin{align}
    u_i\big(\theta_i,b_{-i};\pi,\theta_i\big)\geq u_i\big(b_i,b_{-i};\pi,\theta_i\big).
    \label{eq:dsicsingle}
\end{align}
While \DSIC\ is a strong IC requirement, it is not always attainable and hence, it is common to relax this requirement. In this work, we consider Nash incentive compatibility (\NIC). We say that $\pi$ is \NIC, if 
for all $i\in[n]$, $\theta_i,b_i\in\Theta_i$ and $\theta_{-i}\in\Theta_{-i}$, we have
\begin{align}
    u_i\big(\theta_i,\theta_{-i};\pi,\theta_i\big)\geq u_i\big(b_i,\theta_{-i};\pi,\theta_i\big).
    \label{eq:nicsingle}
\end{align}
That is, given that other agents report their types truthfully ($b_{-i} = \theta_{-i}$), reporting one's type truthfully is an optimal strategy. This means that truth-telling is a Nash equilibrium.
Since we aim to learn a mechanism in online, adversarial environments without assuming type distributions, NIC is appropriate for our setting.
NIC is also stronger than Bayes-Nash incentive compatibility (BNIC), which only guarantees truth-telling in expectation over type distributions~\citep{nisan2007algorithmic}.
%\kkcomment{condensed.}
%Moreover, \NIC\ is stronger than another widely studied IC guarantee called Bayes-Nash incentive compatibility (\BNIC) which guarantees the above \textit{in expectation} over some type distributions~\citep{nisan2007algorithmic}.
%While \NIC\ is weaker than \DSIC, it makes learning tractable in online settings.
% Moreover, \NIC\ is stronger than another widely studied IC guarantee called Bayes-Nash incentive compatibility (\BNIC) which guarantees the above \textit{in expectation} over some type distributions~\citep{nisan2007algorithmic}.
% Since our goal is to learn a mechanism in online, adversarial environments, \NIC\ is appropriate for our setting we do not assume type distributions.

In this work, we study finding an optimal mechanism in a given class of mechanisms $\Pi$. We  say that $\Pi$ is \NIC\ (\DSIC) if all $\pi\in\Pi$ is \NIC\ (\DSIC). 

%%%%%%%%%%%%%%%%%%%%%%%%%%%%%%%%%%%%%%%%%%%%%%%%%%%%%%%%%%%%%%%%%%
\subsection{Differential Privacy}\label{subsec:diff-priv}
Next, we briefly review differential privacy. We refer readers to \citet{dwork2014algorithmic} for a comprehensive overview. 
Let $X$ be an input domain and $X^m$ be a Cartesian product of $X$, representing the space of datasets of size $m$. Let $Y$ be an output space. We define DP as follows. 

\begin{definition}[Differential privacy]
    \label{def:dp}
    A mapping $q:X^m\rightarrow\Delta(Y)$ is $\eta$-differentially private ($\eta$-DP) if for all $A\subseteq Y$ and all inputs $x=(x_1,\dots,x_m)\in X^m$ and $x'=(x'_1,\dots,x'_m)\in X^m$ that differ in at most one entry ($x_j\neq x'_j$ for at most one $j\in[m]$), the following holds:
    \begin{align*}
       \pr_{y\sim q(x)}\big(\,y\in A\,\big)\leq e^\eta\pr_{y\sim q(x')}\big(\,y\in A\,\big).
    \end{align*}
\end{definition}
Intuitively, this means that the output of an $\eta$-DP algorithm does not change significantly if we change one input data point.
In sequential learning settings, we extend this definition as follows.

\begin{definition}[Weak $\eta$-DP sequence]
    Let $q_{\leq T}:=(q_t)_{t\in[T]}$ be a sequence of mappings where $q_t:X^{t-1}\rightarrow\Delta(Y)$. We say that $q_{\leq T}$ is a weak $\eta$-DP sequence if $q_t$ is $\eta$-DP for each $t\in[T]$.
    \label{def:weak-dp-seq}
\end{definition}

Readers familiar with the literature will recognize that the above definition is weaker than the more standard definition (e.g.~\citet{jain2012differentially}), which requires that $\prod^T_{t=1}q_t(\cdot)$, viewed as a map from $X^T$ to $\Delta(Y^T)$ is $\eta$-DP.
While the stronger definition is relevant from a data privacy perspective, any compliant algorithm often comes with inevitable performance degradation.
However, for our purpose of designing an IC online learning method, the weaker extension in Definition~\ref{def:weak-dp-seq} is sufficient.
In particular, while the vanilla Hedge algorithm
is not strongly $\eta$-DP~\citep{jain2012differentially}, in~\S\ref{subsec:learning-via-hedge}, we show that it is weakly $\eta$-DP. In fact, this follows immediately from a well-known DP technique called the exponential mechanism, described below.

\begin{proposition}[Exponential mechanism \cite{mcsherry2007mechanism}]
    Consider a map $q:X^m\rightarrow\Delta(Y)$ such that $q(x)(y)\propto \exp\!\rbra{\eta\,g(x,y)}$, for some $g:X^m\times Y\rightarrow\rr$, for all $x\in X^m$ and $y\in Y$. Let $\Delta g:=\max_{x,x',y}|g(x,y)-g(x',y)|$ where the maximum is taken over $x,x'\in X^m$ that differ in at most one entry. Then $q$ is $2\eta\Delta g$-DP.
    \label{prop:exp-mech}
\end{proposition}
\section{Nash Incentive-compatible Online Mechanism Learning}\label{sec:icoml}
In this section, we describe the problem of Nash incentive-compatible online mechanism learning.
%%%%%%%%%%%%%%%%%%%%%%%%%%%%%%%%%%%%%%%%%%%%%%%%%%%%%%%%%%%%
\subsection{Online Mechanism Learning}
A set of $n$ agents repeatedly interact with an online mechanism over a sequence of $T$ rounds, although not all agents need to appear on all rounds.
On each round, the environment chooses types $\theta_{i,t}$ one for each participating agent and the social planner's objective $G_t$. Simultaneously, the online mechanism $\M$ uses past information, to choose a single-round mechanism $\pi_t$ to collect agent types and choose an outcome $s_t\in S$.
Our goal is to design a learning algorithm that is competitive against the single best mechanism in hindsight in some mechanism class $\Pi$. We have outlined the protocol for online mechanism learning in Protocol~\ref{protocol:1}. Next, we describe the individual components.

\begin{framefloat}[t]
    \textbf{Protocol 1.} Online Mechanism Learning \refstepcounter{protocol}\label{protocol:1}\\
    \vspace{-0.5em}\\
    For each round $t=1,\dots,T$:
    \begin{enumerate}
        \vspace{0.5em}
        \item $\M$ chooses a single-round mechanism $\pi_t$ based on past reported types $b_{<t}$ and observed objective functions $G_{< t}$.
        \vspace{0.5em}
        \item Each participating agent observes $\pi_t$ and reports their types $b_{i,t}\in\Theta_i$.
        \vspace{0.5em}
        \item $\M$ executes $\pi_t$ with input $b_t$ and realizes an outcome $s_t\sim\pi_t(b_t)$.
        \vspace{0.5em}
        \item Each participating agent experiences utility $u_i(\theta_{i,t},s_t)$.
        \vspace{0.5em}
        \item $\M$ observes $G_t$ and experiences $G_t\big(\theta_t,s_t\big)$.
    \end{enumerate}
\end{framefloat}

\textbf{Agents.}
There are $n$ agents, indexed by $i\in[n]$. 
Each agent $i\in[n]$ appears on a subset $\T_i\subset [T]$ of the rounds.
Let $\theta_i:=(\theta_{i,t})_{t\in\T_i}$ where $\theta_{i,t}\in\Theta_i$ be the sequence of types for agent $i$ in the rounds that she participates. 
Let $\gamma_i:[T]\rightarrow\unit$, be an agent-specific monotonically non-increasing function which discounts the agent's future utilities. 
If the mechanism chooses outcomes $(s_t)_{t\in [T]}$ on each round, the agent's utility is
$\sum_{t\in\T_i}\gamma_i(t)u_i(\theta_{i,t},s_t)$.

If an agent participates in round $t$, she reports her type $b_{i,t}\in\Theta_i$, not necessarily truthfully to the mechanism.
If an agent does not participate, we will set the reported and true types as $b_{i,t}=\theta_{i,t}=\perp$.
For what follows, denote $b_t:=(b_{i,t})_{i\in[n]}$, $b_{<t}:=(b_\tau)_{\tau\in[t-1]}$ and $b_{\leq T}:=(b_t)_{t\in[T]}$. Define $\theta_t$, $\theta_{<t}$ and $\theta_{\leq T}$ similarly. Finally, we write the agents as a tuple: 
\begin{align*}
    A:=\rbra{n,T,(\theta_t)_{t\in[T]},(\T_i)_{i\in[n]},(\gamma_i)_{i\in[n]}}   
    \numberthis\label{eq:agents}
\end{align*}

\textbf{Objectives.} The social planner's objective functions are given by $G_t:\prod_{i\in[n]}\!\rbra{\Theta_i\cup\{\perp\}}\!\times\!S\!\rightarrow\! \munit$ for each round $t\in[T]$. These objectives could be application-specific, and be chosen by an oblivious adversary.

\textbf{Online mechanism.}
An online mechanism $\M = (\M_t)_{t\in\nat}$, is a sequence of policies, where $\M_t$ maps past reported types $b_{<t}$ and objectives $G_{<t}$ to a distribution over $\U$, where $\U:=\{\pi:\prod_{i\in[n]}\!\rbra{\Theta_i\cup\{\perp\}}\rightarrow\Delta(S)\}$ is the class of \emph{all} single-round mechanisms.
Some $\pi_t$ is sampled from this distribution $\M_t(b_{<t},G_{<t})$ and executed on round $t$.

\textbf{Regret.}
To define the regret,
first consider the cumulative objective attained by an online mechanism $\M$, given reported types, $b_{\leq T}$ with respect to agents $A$.
\begin{align*}
    \alg\big(b_{\leq T},\M;A\big):=\expec\!\sbra{\sum^T_{t=1}G_t\big(\theta_t,s_t\big)}.
\end{align*}
Here, the expectation is over the sequence of outcomes $s_{\leq T}$ generated by running Protocol \ref{protocol:1} with $\M$, $A$ and $b_{\leq T}$. We compare the above quantity to the optimal cumulative objective $\opt(\Pi;A)$ attainable within a class of mechanisms $\Pi$ with respect to agents $A$:
\begin{align*}
    \opt(\Pi;A)&:=\max_{\pi\in\Pi}\sum^T_{t=1}\expec_{s\sim \pi(\theta_t)}G_t(\theta_t,s)
\end{align*}
The performance of $\M$ is measured by the gap between the above two quantities, which is called the regret:
\begin{align*}
    \reg\big(b_{\leq T},\M,\Pi;A\big):=\opt(\Pi;A)-\alg\big(b_{\leq T},\M;A\big).
\end{align*}

\subsection{Nash Incentive Compatibility}
We next describe the IC requirement for an online learning mechanism.
For this, let $h_{<t}$ denote the history up to round $t-1$.
It consists of the previous single-round mechanisms chosen, submitted \emph{and} true types of all agents, outcomes, and objectives, i.e. $h_{<t}:=(\pi_{<t},b_{<t},\theta_{<t},s_{<t},G_{<t})$.

We will assume that throughout the $T$ rounds, agents are aware of their utility function $u_i$, and their true types $\theta_i = (\theta_{i,t})_{t\in\T_i}$ in past, current, and future rounds.
Moreover, on round $t$, agents may have access to the realized history $h_{<t}$ while not required. Note that while agents may know the true types, the online mechanism does not.

\textbf{Agent's strategy.}
An agent $i$'s strategy $\sigma_i:=(\sigma_{i,t}(\cdot))_{t\in\T_i}$ is a sequence of policies,
where $\sigma_{i,t}$ maps the realized history 
$h_{<t}$ to a reported type.
That is, $b_{i,t} = \sigma_{i,t}(h_{<t})\in\Theta_i$.
If agents follow randomized strategies, our results will hold for all realizations of this strategy.
We use $\sigma:=(\sigma_i)_{i\in[n]}$ to denote a strategy profile of all agents and $\sigma_{-i}$ to denote a strategy profile except agent $i$'s.

The \emph{truth-telling strategy} for agent $i$
reports the true type regardless of history.
With a slight abuse of notation, we denote this as $\theta_i:=(\theta_{i,t}(\cdot))_{t\in\T_i}$ where $\theta_{i,t}(\cdot):=\theta_{i,t}$. We call $\theta:=(\theta_i)_{i\in[n]}$ the truth-telling strategy profile.

\paragraph{Nash incentive-compatible learning.}
To define the NIC requirement, we should define the expected utility for an agent.
Given an online mechanism $\M$, agents $A$, and strategy profile $\sigma$,
with a slight abuse of notation,
we denote the expected utility for agent $i$ as
\begin{equation}
    u_i(\sigma;\M,A):=\expec\!\sbra{\,\sum_{t\in\T_i}\gamma_i(t)u_i(\theta_{i,t},s_t)}.
    \label{eq:expectedutilityonline}
\end{equation}
Here, $u_i$ on the RHS is the utility of the agent with type $\theta_{i,t}$ for outcome $s_t$, and is similar to the RHS of~\eqref{eq:expectedutilitysingle}.
Moreover,  the expectation is over the joint distribution of $s_{\leq T}$, induced by the interaction between the strategies $\sigma$, mechanism $\M$, and agents $A$ via Protocol \ref{protocol:1}.
We can now extend the definition of NIC (see~\eqref{eq:nicsingle}) for online mechanisms.

\begin{definition}[\NIC, Online setting]
    An online mechanism $\M$ is Nash incentive-compatible (\NIC) with respect to agents $A$, if  for every agent $i$,
    \begin{align*}
        u_i(\theta_i,\theta_{-i};\M,A)\geq u_i(\sigma'_i,\theta_{-i};\M,A),
    \end{align*}
    that is, no agent $i$ has incentive to deviate from the truth-telling strategy $\theta_i$, given that other agents are truth-telling. 
    \label{def:nic-online}
\end{definition}
\vspace{-1em}
In other words, NIC guarantees that the truth-telling strategy profile is a Nash equilibrium of the game induced by $\M$ and $A$. Finally, we can state our definition of a Nash incentive-compatible online mechanism.

\begin{definition}[\NICOM]
    Let $\A$ be a class of possible agents (i.e. a set of the tuples in~\eqref{eq:agents}), and let $\A_T\subseteq\A$ be the subset of all $A\in\A$ whose total number of rounds is $T$. We say that a mechanism $\M$ is a Nash incentive-compatible online mechanism (\NICOM) for a single-round mechanism class $\Pi$ with respect to an agent class $\A$, if,
    \begin{enumerate}
        \item The regret of $\M$ is sub-linear given that agents report their types truthfully. That is,
        \begin{align*}
            \sup_{A\in\A_T}\reg(\theta_{\leq T},\M,\Pi;A)\in o(T).
        \end{align*}
        \vspace{-0.2in}
        \item For every $A\in\A$, $\M$ is \NIC\ with respect to $A$.
        %\kkcomment{Why do we have an $\A_T$ above but only a $\A$ here?}
    \end{enumerate}
    \label{def:nicom}
\end{definition}
In the above definition, we assume that the mechanism $\M$ has full knowledge of the agent class $\A$, but not of realized agents $A\in\A$ (e.g. their true types) except for $T$.

\paragraph{Long-sightedness.}
Before we proceed, we discuss one more condition necessary to state our results.
If agents prioritize future gains and are willing to incur initial losses for a longer period, incentive-compatible online learning gets more challenging. 
Hence, prior works in this space have assumed that agents either appear a small number of rounds, i.e. $|\T_i| \in o(T)$, and/or that the discounted utilities grow sublinearly, i.e. $\sum^T_{t=1}\gamma_i(t)\in o(T)$ to circumvent this issue~\citep{liu2018learning,abernethy2019learning}.
%In fact,~\citet{amin2013learning} showed that if either of these does not hold, then it is not possible to design an incentive-compatible posted price mechanism even when there is a single buyer.
In fact,~\citet{amin2013learning}'s Theorem 3 implies that if either of these does not hold, it is impossible to design an incentive-compatible posted price mechanism even when there is a single buyer.
Below, we use a similar characterization of an agent's long-sightedness, which combines both definitions.

\begin{definition}[Long-sighted agents]
    The long-sightedness $\alpha$ of the agents $A$ is,
    \begin{align*}
        \alpha(A):=\max_{i\in[n],\,t\in[T]}\min\!\rbra{\frac{1}{\gamma_i(t)}\sum^T_{\tau=t}\gamma_i(\tau),\,|\T_i|}. \\[-0.25in]
    \end{align*}
    Furthermore, for an agent class $\A$, let $\alpha_T(\A):=\sup_{A\in\A_T}\alpha(A)$, where $\A_T$ is the set of all $A\in\A$ whose total rounds are $T$.
    \label{def:long-sightedness}
\end{definition}

When $\gamma_i(t)=\nu^t$ for all $t$, for some $\nu\in\unito$, we have $\alpha(A)\leq \frac{1}{1-\nu}$, which is the maximum (discounted) utility an agent can achieve. On the other extreme, when $\gamma_i(\cdot)=1$, now $\alpha(A)= \max_{i\in[n]}|\T_i|$ is the maximum number of times an agent appears. We will see that our regret bound increases with $\alpha_T(\A)$.

Some dependence on the long-sightedness is unavoidable. The following proposition makes the impossibility mentioned above precise in our setting, for the case of a single agent who appears on every round and an online mechanism that posts a price, possibly at random, on each round. We prove it in Appendix~\ref{app:proof:prop:lower-bound} by adapting the proof of Theorem~3 of \citet{amin2013learning}.

\begin{restatable}{proposition}{PropLowerBoundNIC}
    Let $n=1$ and $\{0,\frac{1}{2},1\}\subseteq\Theta_1\subseteq[0,1]$. For $p\in[0,1]$, let $\pipp_p$ be the posted-price mechanism that sells the item at price $p$ if the reported type is at least $p$, and does not sell it otherwise; an agent of type $\theta$ has utility $\theta-p$ if she buys at price $p$, and $0$ otherwise. Let the objective $G_t$ be the revenue on every round, and let $\Pi\supseteq\{\pipp_{1/2},\pipp_1\}$. Fix $T\in\nat$ and a non-increasing $\gamma_1:[T]\rightarrow(0,1]$, and let $\A$ be the class of agents $A$ with $n=1$, $\T_1=[T]$, discount function $\gamma_1$, and types $\theta_{1,t}\in\{0,\frac{1}{2},1\}$ for $t\in[T]$. If, on every round, an online mechanism $\M$ chooses a posted-price mechanism or a lottery over posted-price mechanisms, and $\M$ is \NIC\ with respect to every $A\in\A$, then
    \begin{align*}
        \sup_{A\in\A}\reg\big(\theta_{\leq T},\M,\Pi;A\big)&\geq\frac{\alpha_T(\A)}{4}\\
        &=\frac{1}{4}\max_{t\in[T]}\frac{1}{\gamma_1(t)}\sum^T_{\tau=t}\gamma_1(\tau).
    \end{align*}
    \label{prop:lower-bound}
\end{restatable}
For example, if $\gamma_1(\cdot)=1$, then $\alpha_T(\A)=T$, and every such online mechanism has regret at least $T/4$.
\section{Method}\label{sec:method}
In this section, we describe our key methodological contributions.
In~\S\ref{subsec:design-nicom}, we provide a general recipe for building a NICOM from a weak DP sequence.
Then, in~\S\ref{subsec:learning-via-hedge}, we instantiate this recipe via the Hedge algorithm.
% In~\S\ref{subsec:discussion-context-reg}, we discuss previous regret bounds for similar settings to help contextualize our main regret bound.
% In~\S\ref{subsec:discussion-lim-ext}, we discuss some limitations and extensions of our method.
In~\S\ref{subsec:discussion-context-reg} and \ref{subsec:discussion-lim-ext}, we provide a series of discussions on our main result including references to previous bounds for similar settings.
%%%%%%%%%%%%%%%%%%%%%%%%%%%%%%%%%%%%%%%%%%%%%%%%%%%%%%%%%%%%%%%%%%
\subsection{Designing NIC Online Mechanisms}
\label{subsec:design-nicom}
\paragraph{Commitment mechanism.}
A key ingredient of our construction is a
\emph{commitment mechanism} $\picom$.
This refers to any strictly DSIC mechanism--usually far from optimal--that penalizes non-truthful behaviors.
Our method will randomly choose between a mechanism from the target class $\Pi$, and the commitment mechanism.

Below, we define the \textit{penalty gap}, which quantifies the power of a commitment mechanism, i.e. its ability to penalize non-truthful behavior. For this, recall that
$\U$ is the class of all single-round mechanisms.

\begin{definition}[Penalty gap] 
    For a single-round mechanism $\pi\in\U$, the penalty gap $\beta(\pi)$ is the minimum difference in any agent's utility between reporting truthfully and non-truthfully, regardless of the others' behavior. That is,
    \begin{align*}
        \beta(\pi) := \min_{i, \theta_i, b_i\neq \theta_i, b_{-i}}\Bigg[& \expec_{s\sim\pi(\theta_i,b_{-i})}u_i(\theta_i,s) \\
            &\hspace{0.4in} -\; \expec_{s\sim\pi(b_i,b_{-i})}u_i(\theta_i,s)
            \Bigg].
    \end{align*}
    \label{def:penalty-gap}
    \vspace{-0.2in}
\end{definition}

Our method will not require $\picom$ to be in the
target mechanism class $\Pi$.
The purpose of $\picom$ is solely to cancel out any potential gains from lying.
As such, we require a commitment mechanism with $\beta(\picom)>0$, and our regret bound
will scale with how powerful $\beta(\picom)$ is.
As we show in the examples in~\S\ref{sec:examples} and Appendix~\ref{app:resource-alloc}, it is often possible to find a good commitment mechanism for a problem.

\paragraph{NIC via Weak DP Sequences.}
Now, we present a general construction of NICOMs, using a weak differentially private sequence (Definition \ref{def:weak-dp-seq}) and a commitment mechanism.
%Let $\Pi$ be a mechanism class. 

\begin{definition}[Weak $\eta$-DP mechanism]
    Consider any $\lambda\in\unit$, commitment mechanism $\picom$ and weak $\eta$-DP sequence $(\qwdp_t)_{t}$. Here, $\qwdp_t$ is a map from $b_{<t}$ and $G_{<t}$ to a distribution over $\Pi$.
    Let $\Mwdp$ be an online mechanism that outputs the following mechanism in~\eqref{eqn:gen-con} for each $t$. Below, $\pisr_t\in\Pi$ is a single-round mechanism sampled from $\qwdp_t(b_{<t},G_{<t})$.
    We have,
    \begin{align}
        \pi_t(\cdot):=(1-\lambda)\cdot\pisr_t(\cdot)+\lambda\cdot\picom(\cdot).
        \label{eqn:gen-con}
    \end{align}
    \label{def:gen-con}
    \vspace{-0.2in}
\end{definition}
The above definition says that $\Mwdp$ outputs a lottery which realizes the commitment mechanism with probability $\lambda$, and the recommendation of a weak $\eta$-DP sequence $\qwdp_t$ (e.g. an algorithm such as Hedge) with probability $1-\lambda$.

Next, we state our main theoretical result of this subsection.

\begin{restatable}{theorem}{ICviaDP}
    Let $\Pi$ be a \NIC\ single-round mechanism class, $A$ be agents, and $\eta\in(0,1]$. If $\frac{\lambda}{4\eta}\beta(\picom)\geq\alpha(A)$, then $\Mwdp$ is NIC with respect to agents $A$.
    \label{thm:ic-via-dp}
\end{restatable}
From this theorem, we can conclude that when $\beta(\picom)>0$ and the agents are not too long-sighted (i.e. $\alpha_T(\A)$ is sub-linear in $T$), for a suitable choice of $\lambda$ and $\eta$, the proposed mechanism $\Mwdp$ is Nash incentive-compatible.

\textit{Proof sketch:}
The weak $\eta$-DP sequence limits how much the agent can earn by changing a reported type in each round, given that other agents' strategies do not depend on histories. Combining the maximum possible earnings limited by the weak $\eta$-DP sequence and the minimum penalty induced by the commitment mechanism, we show that no agent has an incentive to deviate from truth-telling in the last round when other agents are truthful. By recursively applying this argument, we demonstrate that being truthful in all rounds is optimal for any agent whenever others are truthful, making truth-telling a Nash equilibrium.
%reporting types truthfully is optimal given that other agents report their types truthfully. Hence, the truth-telling is a Nash equilibrium.
The full proof is in App.\ref{app:proof:thm:ic-via-dp}. 

\subsection{NIC online learning via Hedge}
\label{subsec:learning-via-hedge}
To apply Theorem~\ref{thm:ic-via-dp} and design a \NICOM, we need to find a sequence of maps (online learners) $\qwdp_{\leq T}$ that satisfies the weak DP constraint, while also 
achieving sub-linear regret. In this subsection, we show that the Hedge algorithm~\citep{hazan2016introduction}, one of the most common online learning algorithms for the expert problem, satisfies both properties.
For simplicity of exposition, we will assume that the mechanism class $\Pi$ is finite. If this is not the case,
we can discretize the mechanism's learnable parameters, and then separately control the discretization error.

\paragraph{Hedge algorithm.}
First, we state the Hedge update rule. Here, we are interested in finding a mechanism $\pi$ in a finite set of mechanisms $\Pi$, which maximizes the sum of objectives $\sum_t G_t(\theta, s)$.
For any $\eta\geq 0$, past reported types $b_{<t}$, observed objectives $G_{<t}$ and a mechanism class $\Pi$, let $\qhedge_{\eta,t}(b_{<t},G_{<t})$ be the probability distribution over $\Pi$ whose probability masses are given by
 \begin{align}
     \qhedge_{\eta,t}(b_{<t},G_{<t})(\pi)\propto\exp\!\rbra{\eta\sum^{t-1}_{\tau=1}F_\tau(b_\tau,\pi)}.
     \label{def:hedge-update}
\end{align}
Here $F_t(\theta_t,\pi):=\expec_{s\sim\pi(\theta_t)}G_t(\theta_t,s)$.
Consequently, we may view each $\qhedge_{\eta,t}(b_{<t},G_{<t})$ as an exponential mechanism (Proposition \ref{prop:exp-mech}) with $g:=\sum^{t-1}_{\tau=1}F_\tau$ and $\Delta g=2$ as $F_t(\cdot,\cdot)\in\munit$. Hence, we immediately conclude that:

\begin{restatable}{lemma}{LemHedgeDP}
    $\big(\qhedge_{\eta,t}\big)_{t\in[T]}$ given in~\eqref{def:hedge-update} is a weak $4\eta$-DP sequence.
    \label{lem:hedgedp}
\end{restatable}

\textbf{The Hedge \NICOM:} We can now present our main method, an algorithm for online mechanism learning using Hedge. For any $\lambda\in\unit$, a commitment mechanism $\picom$, and mechanism class $\Pi$, let $\Mhedge$ be an online mechanism that outputs the following mechanism for each $t\in[T]$. Below, $\pisr_t\in\Pi$ is a single-round mechanism sampled from $\qhedge_{\eta,t}(b_{<t},G_{<t})$. That is,
\begin{align}
    \pi_t(\cdot):=(1-\lambda)\cdot\pisr_t(\cdot)+\lambda\cdot\picom(\cdot).
\label{eqn:exp-learner}
\end{align}
Since $\big(\qhedge_{\eta,t}\big)_{t\in[T]}$ is a weak $4\eta$-DP sequence, Theorem~\ref{thm:ic-via-dp} applies to $\Mhedge$ with $4\eta$ in place of $\eta$.

Our theorem below states that $\Mhedge$ is a \NICOM\ with suitable choices of $\lambda$, $\eta$ and $\picom$. 

\begin{restatable}{theorem}{HedgeLearner}
    Let $\Pi$ be a \NIC\ mechanism class and $\A$ be an agent class with $\alpha_T(\A)\in o(T)$. Then, the Hedge learner $\Mhedge$ \eqref{eqn:exp-learner} is a \NICOM\ with regret guarantee
    \begin{align*}
        \sup_{A\in\A_T}\reg(&\theta_{\leq T},\,\Mhedge,\Pi;A)\\
        &\in O\rbra{\big(\log|\Pi|+\beta(\picom)^{-1} \big)\cdot\sqrt{\alpha_T(\A)\,T}}
    \end{align*}    
    for $\eta:=\min\cbra{1/\sqrt{\alpha_T(\A)\,T},\ \beta(\picom)/(16\alpha_T(\A))}$ and $\lambda:=16\eta\,\alpha_T(\A)/\beta(\picom)$. In particular, when $\alpha_T(\A)\in O(T^h)$, the regret is in $O\big(T^{\frac{1+h}{2}}\big)$.
    \label{thm:hedge-mech}
\end{restatable}
\vspace{-0.75em}
\textit{Proof sketch:} The proof follows almost directly from the standard regret bound of the Hedge algorithm \cite{hazan2016introduction}. By the linearity of expectation, the commitment mechanism causes a small $2\lambda T$ offset to the standard regret. After choosing appropriate parameters, we obtain the stated regret. We give the full proof of Theorem \ref{thm:hedge-mech} in Appendix \ref{app:proof:thm:hedge-mech}.

As mentioned in~\S\ref{subsec:contrib}, thanks to Theorem \ref{thm:ic-via-dp}, which relies only on our weak notion of sequential DP, we are able to utilize the Hedge algorithm, resulting in a logarithmic dependence on the size of the mechanism class $|\Pi|$.
%%%%%%%%%%%%%%%%%%%%%%%%%%%%%%%%%%%%%%%%%%%%%%%%%%%%%%%%%%%%%%%%%%
\subsection{Discussion I: Contextualizing the Regret Bound}
\label{subsec:discussion-context-reg}
In this subsection, we provide several pointers to help interpret and contextualize our main regret bound in Theorem \ref{thm:hedge-mech}. It is important to beware that the regret bounds cited below, which have been suitably converted to meet our convention, are \textit{not} directly comparable to ours due to differences in problem settings and strategy-proofness guarantees.

The work most closely related to ours is by \citet{abernethy2019learning}, who present an approximately IC online auction within the Bayesian framework. They establish a regret bound of $O(T/\epsilon+(k\epsilon)^{1/3}T)$ with an IC guarantee, where $\epsilon$ is the DP parameter that appears in the approximate IC guarantee and $k$ is the maximum number of rounds an agent can participate in. This bound becomes vacuous when $k$ grows with $T$, leading to a superlinear regret even when $\epsilon$ tends to $0$. Specifically, if we formulate this bound in terms of our $h$ parameter so that $k\in O(T^h)$, the bound becomes $O(T^{1+h/3})$ for fixed $\epsilon$. In contrast, our bound remains sublinear $O(T^{(1+h)/2})$ whenever $k$ is sublinear.

Other relevant works by \citet{amin2013learning} and \citet{liu2018learning} consider \textit{strategic regret} in the online posted-price auction, which is a regret measured with respect to non-myopic buyers and is defined as $\textsf{SReg}(\textsf{ALG}):=\textsf{Rev}(p^\star,v)-\textsf{Rev}(\textsf{ALG},B)$, where $\textsf{Rev}(p^\star,v)$ is the revenue by the best-fixed price $p^\star$ with respect to the buyer's true values $v:=(v_1,...,v_T)$ and $\textsf{Rev}(\textsf{ALG},B)$ is the revenue of algorithm $\textsf{ALG}$ with respect to non-myopic buyer $B$.

For $T^h$ long-sightedness (i.e., effective horizon), \citet{amin2013learning} offer strategic regret bounds of $O(T^{1/2+h})$ when the buyer's valuation is fixed in all rounds (Theorem 1 therein), which is slightly worse than our $O(T^{(1+h)/2})$ bound, and $\widetilde{O}(c T^{\sqrt{h}})$ for stochastic valuations, where $c$ is an instance-dependent factor (Theorem 2 therein). However, these settings are less general than our adversarial setting.

\citet{liu2018learning} demonstrate a regret bound of $\widetilde{O}(T^{(3+h)/4})$ for adversarial settings when there are continuous values, but their analysis can obtain a $\widetilde{O}(T^{(1+h)/2})$ bound for discrete values which match our regret bound. Note that their algorithm is specific to single-item auctions, and conceptually more complicated than ours.

As seen above, the recurring theme of regret bounds that are ``robust'' against non-myopic agents is their dependency on the agents' long-sightedness. Notably, all of the above bounds become vacuous when the long-sightedness becomes linear, that is, when $h=1$. This phenomenon aligns with the negative result given in \citet{amin2013learning}.

%%%%%%%%%%%%%%%%%%%%%%%%%%%%%%%%%%%%%%%%%%%%%%%%%%%%%%%%%%%%%%%%%%
\subsection{Discussion II: Limitations and Extensions}
\label{subsec:discussion-lim-ext}

Our construction relies on the commitment mechanism and our regret bound depends on the power of the commitment mechanism to penalize untruthful agents. Our examples in~\S\ref{sec:examples} and Appendix~\ref{app:resource-alloc} show that it is often possible to find such a commitment mechanism, quantify its penalty gap, and obtain transparent bounds on the regret. That said, designing a good commitment mechanism may not always be straightforward for a problem.

To simplify the exposition, we only consider \textit{exact} commitment mechanisms that strictly punish any type report that is not \textit{equal} to the true type.
The existence of an exact commitment mechanism often relies on the assumption that agents' types are discrete (for continuous types, the gap may be arbitrarily close to zero). Hence, in this work, we only consider mechanisms with discrete agent types. Our construction can be extended to incorporate continuous types via an approximate commitment mechanism, leading to an approximate IC guarantee.
However, to simplify the presentation and focus on the main insights, we focus on discrete types and consider \text{exact} commitment mechanisms.

In mechanism design, we are often interested in \emph{individual rationality} (IR), which stipulates that the agent's utility when participating truthfully is larger than when not participating.
While we have not explicitly considered IR, by choosing an IR mechanism class and 
%an IR 
commitment mechanism, we straightforwardly obtain an IR online mechanism.

We also reiterate that we assumed $|\Pi|$ was finite in~\S\ref{subsec:learning-via-hedge} when applying the Hedge algorithm. If this is not the case, we can choose a discretization with good coverage, apply Hedge on the discretization, and then control the discretization error to obtain final regret bounds.

Finally, we also wish to mention that our methods can be quite straightforwardly extended to contextual settings, where an online mechanism takes in a piece of information $x\in X$ (such as UserID) as an auxiliary input to choose a mechanism for each round. In this case, we can simply instantiate the Hedge over a subset of policies $\Pi^X$.
\section{Application Examples}\label{sec:examples}
We now apply our framework to develop no-regret NICOMs for three different mechanism design problems:
\emph{(i)} Online facility location,
\emph{(ii)} VCG auctions with reserve prices under \textit{ex-post} externality, and
\emph{(iii)} online optimal resource allocation.
Of these, \emph{(i)} and \emph{(iii)} are examples of mechanism design without money.
%Due to space constraints, we present the facility location problem here, and defer the other two to Appendix~\ref{app:additional-ex}.
Due to space constraints, we only present the first two problems here, and defer the online optimal resource allocation to Appendix~\ref{app:resource-alloc}.

\subsection{Online Facility Location}
Consider a government trying to locate $k\geq 2$ moving facilities (e.g. moving clinics, libraries) every day. The local government wants to locate these facilities to minimize the average distance between each resident and their closest facility, weighted by utilization rate. To simplify the problem, we assume $n$ residents are located on a finite $1$d grid $I_m:=\cbra{0,\frac{1}{m},\frac{2}{m},\dots,1}$, hence $\Theta_i=I_m$ for all $i\in[n]$. A social outcome $s\in S$ is a tuple $((x_1,o_1),\dots,(x_k,o_k))$ where $x_l\in I_m$ and $o_l\subseteq [n]$, where $x_l$ indicates the location of facility $l$ and $o_l$ is the set of residents who are allowed to use facility $l$.

Then, for an outcome $s$, resident $i$'s utility on round $t$ when their location is $\theta_{i,t}$  is
\begin{align*}
    u_i(\theta_{i,t},s):=1-\min_{l\in [k]\,\text{s.t.}\,i\in o_l}|x_l-\theta_{i,t}|,
\end{align*}
and $u_i(\theta_{i,t},s):=0$ if $i\notin o_l$ for any $l\in[k]$. The government's objective on round $t$ given types $\theta_t$ is 
\begin{align*}
    G_t(\theta_t,s):=\frac{1}{n}\sum_{i\in[n]}r_{i,t}\cdot u_i(\theta_{i,t},s).
\end{align*}
Here, $r_{i,t}\in\unit$ is the utilization rate of a facility by the resident $i$ during round $t$.

\paragraph{Posted-location mechanism.}
We consider the class of ``posted-location'' mechanisms, which is parametrized by $w\in I^k_m$ representing the locations of facilities. This class $\Pipl_{m,k}$ is defined as follows: for each $w\in I^k_m$, $\pipl_w\in\Pipl_{m,k}$ is a mechanism that deterministically outputs $\rbra{(w_1,[n]),\dots,(w_k,[n])}\in S$ regardless of input types. That is, it deterministically sets the locations of the facilities at $w\in I^k_m$, and allows any agents to use any facility so that agents can pick the closest facility to them. As the mechanism's output does not depend on the reported types, it is \NIC{} in a single round; but if we naively use previous reports to set future locations,  the online mechanism is not \NIC.

\paragraph{Commitment mechanism.}
We construct a commitment mechanism $\picom$ for the online facility locating as follows: Given type reports $b_t:=(b_{i,t})_{i\in[n]}$, $\picom$ samples $j\in[m]$ uniformly at random and places one facility in $\frac{j-1}{m}$ and other $k-1$ facilities in $\frac{j}{m}$ (i.e. $x_1=\frac{j-1}{m}$ and $x_l=\frac{j}{m}$ for $2\leq l \leq k$). Then, $\picom$ sets $o_l=\cbra{i\in[n]:b_{i,t}\neq\perp,\ l\in\argmin_{l'\in[k]}|x_{l'}-b_{i,t}|}$ for each $l\in[k]$, that is, facility $l$ is only accessible to residents whose closest facility is $l$ according to their reports.

Then, the penalty gap $\beta(\picom)$ is at least $m^{-2}$.
To see this, suppose $b_{i,t}>\theta_{i,t}$. If $\picom$ chooses an outcome $s$ such that $x_1=\theta_{i,t}$, the resident $i$ suffers at least $m^{-1}$ utility loss compared to that of reporting $\theta_{i,t}$, as they have to use some facility $l\geq 2$ located at $x_l=\theta_{i,t}+m^{-1}$. As $\picom$ chooses such an outcome with probability at least $m^{-1}$, the expected utility loss is at least $m^{-2}$. By a similar argument, the same holds when $b_{i,t}<\theta_{i,t}$. Hence $\beta(\picom)\geq m^{-2}$.

\paragraph{Regret bound.}
Finally, as there are $|I^k_m|$ possible placements of facilities, we have that $\big|\Pipl_{m,k}\big|=|I_m|^k=(m+1)^k$. Hence, there is a \NICOM\ for $\Pipl_{m,k}$ with regret $O\!\rbra{ (k\log m + m^2)\cdot T^{\frac{1+h}{2}}}$ for any agent class $\A$ such that $\sup_{A\in\A_T}\alpha(A)\in O(T^h)$ for some $h\in\unito$. 

\subsection{VCG Mechanism under \textit{Ex-post} Externalities}
Consider a government selling a permit to $n$ firms via VCG mechanism \cite{nisan2007algorithmic}. We use $o\subseteq[n]$ to denote the set of firms who get the permit and $p\in\unit^n$ is the price of the permit charged to each firm. Then, the social outcome space is $S:=\cbra{(o,p):o\subseteq [n],p\in\unit^n}$. For all $i\in[n]$, let type space $\Theta_i$ be $I_m:=\cbra{0,\frac{1}{m},\frac{2}{m},\dots,1}$, representing a firm's possible valuations of the permit.

The government's objective function is the average social welfare: $G_t(\theta_t,s):=\frac{1}{n}\big(\sum_{i\in o}\theta_{i,t} - c_t(o)\big)$ where $\theta_{i,t}$ is firm $i$'s valuation of the permit at round $t$ and $c_t(\cdot)\in[0,n]$ is an \textit{ex-post} externality (such as pollution) revealed after round $t$. To balance the total value and externalities, the government introduces bidder-specific reserve prices $w\in I_m^n$, where $w_i$ is the minimum price of the permit charged to firm $i$.

\paragraph{VCG mechanism with bidder-specific reserve prices.} We define a class of mechanisms $\Pivcg_{n,m}$ as follows: Given input bids $b:=(b_i)_{i\in[n]}$, $\pivcg_w\in\Pivcg_{n,m}$ with bidder-specific reserve prices $w\in I_m^n$ decides the set of winners $o^\star\subseteq[n]$ as follows:
\begin{align*}
    o^\star\in\argmax_{o\subseteq\P_w(b)}\sum_{j\in\P_w(b)}b_j\cdot\ind\big[j\in o\big],
\end{align*}
where $\P_w(b):=\cbra{i\in[n]:b_i\neq\perp,\,b_i\geq w_i}$, and charges each winner $i\in o^\star$ the price $\max\{w_i,p^\star_i\}$ and every other firm nothing, where
\begin{align*}
    p^\star_i&=\max_{o\subseteq\P_w(b)}\bigg(\sum_{j\in\P_w(b)\setminus\{i\}}b_j\cdot\ind\big[j\in o\big]\bigg)\\
    &\hspace{0.9in}-\sum_{j\in\P_w(b)\setminus\{i\}}b_j\cdot\ind\big[j\in o^\star\big].
\end{align*}
For any reserve price, this mechanism is \NIC\ (actually, \DSIC) with respect to the quasi-linear utility $u_i(\theta_i,s):=(\theta_i-p_i)\cdot\ind[i\in o]$ for $s\in S$ \cite{nisan2007algorithmic}. 

\paragraph{Commitment mechanism.}
Our choice of $\picom$ for this problem is as follows: Given bids $b_t:=(b_{i,t})_{i\in[n]}$, $\picom$ samples $z\in I_{2m}$ uniformly at random and outputs $(o,p)$ where $o=\cbra{i\in[n]:b_{i,t}\neq\perp,\,b_{i,t}\geq z}$ and $p_i=z$ for all $i\in[n]$. Then, the penalty gap $\beta(\picom)$ is at least $(2m(2m+1))^{-1}$ which can be seen as follows: Suppose $b_{i,t}<\theta_{i,t}$. If $\picom$ chooses an outcome such that $p_i=\theta_{i,t}-(2m)^{-1}$, the firm $i$ suffers at least $(2m)^{-1}$ utility loss compared to that of reporting $\theta_{i,t}$. Similarly, the firm $i$ suffers at least $(2m)^{-1}$ when $b_{i,t}>\theta_{i,t}$, if $\picom$ chooses an outcome such that $p_i=\theta_{i,t}+(2m)^{-1}$. As each of these events happens with probability $(2m+1)^{-1}$, $\beta(\picom)\geq(2m(2m+1))^{-1}$.

\paragraph{Regret bound.}
To sum up, we have that  $\beta(\picom)\geq(2m(2m+1))^{-1}$ and $|\Pivcg_{n,m}|=|I_m^n|=(m+1)^n$. Therefore, by Theorem \ref{thm:hedge-mech}, there is a \NICOM\ for this problem with regret $O\!\rbra{(n\log m + m^2)\cdot T^{\frac{1+h}{2}}}$ under any agent class $\A$ such that $\alpha_T(\A)\in O(T^h)$ for some $h\in\unito$.
\section{Summary}\label{sec:conclusion}
We studied an online mechanism design setting, where we wish to design an online mechanism that is Nash incentive-compatible (NIC) throughout the entire learning process against non-myopic agents, while maximizing an application-specific objective.
In contrast with prior work, our method applies to general mechanism design problems in online adversarial settings. We show that a weak notion of DP, which is satisfied by the Hedge algorithm, is sufficient to achieve NIC. This allows us to obtain regret bounds that scale logarithmically in the size of the mechanism class.
We apply our framework to design NIC online learning procedures for three different mechanism design problems.
% \newpage
% \section*{Acknowledgements}
% We would like to thank the reviewers for taking the time and effort necessary to review the manuscript. We sincerely appreciate their valuable comments and suggestions, which helped improve the quality of the manuscript.

% \section*{Impact Statement}
% This paper presents work whose goal is to advance the field of machine learning and mechanism Design. There are many potential societal consequences of our work, none of which we feel must be specifically highlighted here.

% In the unusual situation where you want a paper to appear in the
% references without citing it in the main text, use \nocite
% \nocite{langley00}

\bibliography{ref}
\bibliographystyle{icml2024}

%%%%%%%%%%%%%%%%%%%%%%%%%%%%%%%%%%%%%%%%%%%%%%%%%%%%%%%%%%%%%%%%%%%%%%%%%%%%%%%
%%%%%%%%%%%%%%%%%%%%%%%%%%%%%%%%%%%%%%%%%%%%%%%%%%%%%%%%%%%%%%%%%%%%%%%%%%%%%%%
% APPENDIX
%%%%%%%%%%%%%%%%%%%%%%%%%%%%%%%%%%%%%%%%%%%%%%%%%%%%%%%%%%%%%%%%%%%%%%%%%%%%%%%
%%%%%%%%%%%%%%%%%%%%%%%%%%%%%%%%%%%%%%%%%%%%%%%%%%%%%%%%%%%%%%%%%%%%%%%%%%%%%%%
\newpage
\appendix
\onecolumn
\numberwithin{equation}{section}
\section{Proof of Theorem \ref{thm:ic-via-dp}}\label{app:proof:thm:ic-via-dp}
%%%%%%%%%%%%%%%%%%%%%%%%%%%%%%%%%%%%%%%%%%%
In this appendix, we give a proof of Theorem \ref{thm:ic-via-dp} restated below.

\ICviaDP*

\begin{proof}
    Fix any $i\in[n]$. For any strategy profile $(\sigma_i,\sigma_{-i})$, we can rewrite agent $i$'s expected utility as follows:
    \begin{align*}
        u_i\!\rbra{\sigma_i,\sigma_{-i};\Mwdp,A}&:=\expec\!\sbra{\,\sum_{t\in\T_i}\gamma_i(t)u_i(\theta_{i,t},s_t)}\\
        &=\sum_{t\in\T_i}\gamma_i(t)\expec_{b_{\leq t}}\!\bigg[\expec_{s_t}\!\Big[u_i(\theta_{i,t},s_t)\,\Big|\,b_{\leq t}\Big]\bigg]\\
        &=\expec_{b_{\leq T}}\!\sbra{\,\sum_{t\in\T_i}\gamma_i(t)\expec_{s_t}\!\Big[u_i(\theta_{i,t},s_t)\,\Big|\,b_{\leq t}\Big]},
        \numberthis\label{eq:ui-expec}
    \end{align*}
    where the first expectation is over all randomness and $\expec_{b_{\leq t}}$ $\rbra{\expec_{b_{\leq T}}}$ is the expectation with respect to the distribution of $b_{\leq t}$ $\rbra{b_{\leq T}}$ induced by the interaction between agents $A$ employing strategy $(\sigma_i,\sigma_{-i})$ and online mechanism $\Mwdp$ via Protocol \ref{protocol:1}. And $\expec_{s_t}[\,\cdot\,|b_{\leq t}]$ is the expectation with respect to the following distribution:
    \begin{align*}
        \D(b_{\leq t}):=(1-\lambda)\cdot\qwdp_t(b_{<t},G_{<t})(b_t)+\lambda\cdot\picom(b_t)\in\Delta(S),
    \end{align*}
    where $\qwdp_t(b_{<t},G_{<t})(b_t)\in\Delta(S)$ is the distribution defined by process $s\sim\pi(b_t),\,\pi\sim\qwdp_t(b_{<t},G_{<t})$.
    
    For each $t\in[T]$, let $b_{i,\leq t}:=(b_{i,\tau})_{\tau\in[t]}$ where $b_{i,\tau}:=\perp$, if and only if $\tau\notin\T_i$ and $b_{-i,\leq t}:=(b_{j,\leq t})_{j\in[n]\setminus\{i\}}$, define $\theta_{i,\leq t}$, $\theta_{-i,\leq t}$ similarly but with true types. With these notations, \eqref{eq:ui-expec} can be rewritten as
    \begin{align*}
        u_i\!\rbra{\sigma_i,\sigma_{-i};\Mwdp,A}=\expec_{b_{i,\leq T},\,b_{-i,\leq T}}\!\sbra{\,\sum_{t\in\T_i}\gamma_i(t)\expec_{s_t}\!\Big[u_i(\theta_{i,t},s_t)\,\Big|\,b_{i,\leq t},\,b_{-i,\leq t}\Big]}.
    \end{align*}
    Suppose other players are truthful, that is, $\sigma_{-i}=\theta_{-i}$. Substituting $\sigma_{-i}=\theta_{-i}$ to the above, we have that $b_{-i,\leq T}=\theta_{-i,\leq T}$ and hence
    \begin{align*}
        u_i\!\rbra{\sigma_i,\theta_{-i};\Mwdp,A}=\expec_{b_{i,\leq T}}\!\sbra{\,\sum_{t\in\T_i}\gamma_i(t)\expec_{s_t}\!\Big[u_i(\theta_{i,t},s_t)\,\Big|\,b_{i,\leq t},\,\theta_{-i,\leq t}\Big]\ \middle|\ \theta_{-i,\leq T}},
        \numberthis\label{eq:iterated-expec}
    \end{align*}
    where $\expec_{b_{i,\leq T}}$ is the distribution over agent $i$'s reported types $(b_{i,t})_{t\in\T_i}$ induced by playing $\sigma_i$ while other agents' strategies are $\theta_{-i}$. Using H{\"o}lder's inequality to \eqref{eq:iterated-expec}, we have that
    \begin{align*}
        u_i\!\rbra{\sigma_i,\theta_{-i};\Mwdp,A}\leq\sum_{t\in\T_i}\gamma_i(t)\expec_{s_t}\!\Big[u_i(\theta_{i,t},s_t)\,\Big|\,\hat{b}_{i,\leq t},\,\theta_{-i,\leq t}\Big]
        \numberthis\label{eq:u_i-holder}
    \end{align*}
    for some $\hat{b}_{i,\leq T}$. 
    
    Now, we show that the above upper bound is attained by $\hat{b}_{i,\leq T}$ such that $\hat{b}_{i,t}=\theta_{i,t}$ for all $t\in\T_i$. To see this, define a strategy modification function $\phi_i:\prod^T_{t=1}(\Theta_i\cup\{\perp\})\rightarrow\prod^T_{t=1}(\Theta_i\cup\{\perp\})$ as follows:
    \begin{enumerate}
        \item Given input $b_{i,\leq T}$, $\phi_i$ checks whether there exists $t\in\T_i$ such that $b_{i,t}\neq\theta_{i,t}$.
        \item If there is no such $t\in[T]$, return unmodified $(b_{i,t})_{t\in\T_i}$.
        \item Otherwise, let $t^\star$ be the largest $t\in\T_i$ such that $b_{i,t}\neq\theta_{i,t}$, and return the modified input $(b_{i,t})_{t\in\T_i}$ where $b_{i,t^\star}$ is set to $\theta_{i,t^\star}$ and other elements are unchanged.
    \end{enumerate}
    Let $\phi_i(b_{i,\leq t})$ denote the restriction of $\phi_i(b_{i,\leq T})$ up to round $t$. Then, we claim that for any $b_{i,\leq T}$, the following holds.
    \begin{align*}
      \sum_{t\in\T_i}\gamma_i(t)\expec_{s_t}\!\Big[u_i(\theta_{i,t},s_t)\,\Big|\,b_{i,\leq t},\,\theta_{-i,\leq t}\Big]\leq\sum_{t\in\T_i}\gamma_i(t)\expec_{s_t}\!\Big[u_i(\theta_{i,t},s_t)\,\Big|\,\phi_i\!\rbra{b_{i,\leq t}},\,\theta_{-i,\leq t}\Big].
      \numberthis\label{eq:fi-claim}
    \end{align*}

    \textbf{Proof of \eqref{eq:fi-claim}.} First, it suffices to consider only $b_{i,\leq T}$ such that $b_{i,t}\neq\theta_{i,t}$ for some $t\in\T_i$ as otherwise $\phi_i(b_{i,\leq T})=b_{i,\leq T}$. Moreover, as $\phi_i$ only modifies $b_{i,t^\star}$ where $t^\star$ is the largest $t\in\T_i$ such that $b_{i,t}\neq\theta_{i,t}$, it suffices to show
    \begin{align*}
      \sum_{t\in\T_i\cap[t^\star:T]}\gamma_i(t)\expec_{s_t}\!\Big[u_i(\theta_{i,t},s_t)\,\Big|\,b_{i,\leq t},\,\theta_{-i,\leq t}\Big]\leq\sum_{t\in\T_i\cap[t^\star:T]}\gamma_i(t)\expec_{s_t}\!\Big[u_i(\theta_{i,t},s_t)\,\Big|\,\phi_i\!\rbra{b_{i,\leq t}}\!,\,\theta_{-i,\leq t}\Big],
    \end{align*}
    where $[t^\star:T]:=\cbra{t^\star,t^\star+1,\dots,T}$. By the linearity of expectation, $\sum_{t\in\T_i\cap[t^\star:T]}\gamma_i(t)\expec_{s_t}\!\big[u_i(\theta_{i,t},s_t)\big|b_{i,\leq t},\theta_{-i,\leq t}\big]$ can be unfolded as
    \begin{align*}
        (1-\lambda)\sum_{t\in\T_i\cap[t^\star:T]}\gamma_i(t)\expec_{s_t\sim P_t(b_{i,\leq t})}u_i(\theta_{i,t},s_t)+\lambda\sum_{t\in\T_i\cap[t^\star:T]}\gamma_i(t)\expec_{s_t\sim Q_t(b_{i,t})}u_i(\theta_{i,t},s_t),
        \numberthis\label{eq:assemble-1}
    \end{align*}
    where the former expectations are over $P_t(b_{i,\leq t}):=\qwdp_t\big((b_{i,<t},\theta_{-i,<t}),G_{<t}\big)\big(b_{i,t},\theta_{-i,t}\big)$ and the latter expectations are over $Q_t(b_{i,t}):=\picom\big(b_{i,t},\theta_{-i,t}\big)$.

    Since every mechanism in $\Pi$ is \NIC\ and $\phi_i$ modifies $b_{i,t^\star}$ to $\theta_{i,t^\star}$, the following holds for $t=t^\star$.
    \begin{align*}
        \expec_{s_t\sim P_{t^\star}(b_{i,\leq t^\star})}u_i(\theta_{i,t^\star},\,s_{t^\star})\leq\expec_{s_t\sim P_{t^\star}(\phi_i(b_{i,\leq t^\star}))}u_i(\theta_{i,t^\star},\,s_{t^\star}).
        \numberthis\label{eq:assemble-2}
    \end{align*}
    
    On the other hand, for each $t>t^\star$, $P_t$ is $\eta$-DP with respect to $b_{i,<t}$ by the post-processing lemma (Lemma \ref{lem:post-dp}). Hence, for all $t>t^\star$ we obtain
    \begin{align*}
        \expec_{s_t\sim P_t(b_{i,\leq t})}\!\Big[\,u_i(\theta_{i,t},s_t)+1\,\Big]\leq e^\eta\cdot \expec_{s_t\sim P_t(\phi_i(b_{i,\leq t}))}\!\Big[\,u_i(\theta_{i,t},s_t)+1\,\Big],
    \end{align*}
    which follows from the definition of the weak $\eta$-DP sequence (Definition \ref{def:weak-dp-seq}), the identity $\expec[X]=\int^\infty_0\pr[X\geq x]dx$ for any non-negative random variable $X$ and $u_i(\cdot,\cdot)\in\munit$. As $e^\eta\leq 1+ 2\eta$ for $\forall \eta\in[0,1]$, the above implies
    \begin{align*}
        \expec_{s_t\sim P_t(b_{i,\leq t})}u_i(\theta_{i,t},s_t)\leq  \expec_{s_t\sim P_t(\phi_i(b_{i,\leq t}))}u_i(\theta_{i,t},s_t)+4\eta,
        \numberthis\label{eq:assemble-3}
    \end{align*}
    for all $t\in\T_i\cap[t^\star+1:T]$.

    Next, we consider the latter expectations over $Q_t$ in \eqref{eq:assemble-1}. Since $Q_t(b_{i,t}):=\picom\big(b_{i,t},\theta_{-i,t}\big)$ does not depend on past reported types $b_{i,<t}$, for all $t\in\T_i\cap[t^\star+1:T]$, we have that
    \begin{align*}
        \expec_{s_t\sim Q_t(b_{i,t})}u_i(\theta_{i,t},s_t)=\expec_{s_t\sim Q_t(b_{i,t})}u_i(\theta_{i,t},s_t).
        \numberthis\label{eq:assemble-4}
    \end{align*}
    On the other hand, for $t=t^\star$ term, by the definition of $\beta(\picom)$ (Definition \ref{def:penalty-gap}), we obtain
    \begin{align*}
        \expec_{s_t\sim Q_{t^\star}(b_{i,t^\star})}u_i(\theta_{i,t^\star},s_{t^\star})\leq\expec_{s_t\sim Q_{t^\star}(\theta_{i,t^\star})}u_i(\theta_{i,t^\star},s_{t^\star})-\beta(\picom),
        \numberthis\label{eq:assemble-5}
    \end{align*}
    as $\phi_i$ modifies $b_{i,t^\star}$ to $\theta_{i,t^\star}$.

    Finally, by combining \eqref{eq:assemble-1} -- \eqref{eq:assemble-5}, we obtain
    \begin{align*}
      \sum_{t\in\T_i\cap[t^\star:T]}\gamma_i(t)\expec_{s_t}\!\Big[u_i(\theta_{i,t},s_t)\,\Big|\,\phi_i\!\rbra{b_{i,\leq t}}\!,\,\theta_{-i,\leq t}\Big]&-\sum_{t\in\T_i\cap[t^\star:T]}\gamma_i(t)\expec_{s_t}\!\Big[u_i(\theta_{i,t},s_t)\,\Big|\,b_{i,\leq t},\,\theta_{-i,\leq t}\Big]\\
      &\geq\gamma_i(t^\star)\cdot\lambda\cdot\beta(\picom)-4\eta\cdot\sum_{t\in\T_i\cap[t^\star+1:T]}\gamma_i(t)\\
      &\geq\gamma_i(t^\star)\cdot\lambda\cdot\beta(\picom)-4\eta\cdot\min\!\rbra{\sum^T_{t=t^\star}\gamma_i(t),\,\gamma_i(t^\star)\cdot|\T_i|}\tag{by the monotonicity of $\gamma_i(\cdot)$}\\
      &\geq 4\eta\cdot\gamma_i(t^\star)\cdot\rbra{\frac{\lambda\cdot\beta(\picom)}{4\eta}-\alpha(A)}\\
      &\geq 0,
    \end{align*}
    as $\lambda\beta(\picom)/4\eta\geq\alpha(A)$. Hence, we have \eqref{eq:fi-claim}.\quad\qedsymbol

    Now, we apply \eqref{eq:fi-claim} recursively to itself: 
    \begin{align*}
        \sum_{t\in\T_i}\gamma_i(t)\expec_{s_t}\!\Big[u_i(\theta_{i,t},s_t)\,\Big|\,b_{i,\leq t},\,\theta_{-i,\leq t}\Big]
        &\leq\sum_{t\in\T_i}\gamma_i(t)\expec_{s_t}\!\Big[u_i(\theta_{i,t},s_t)\,\Big|\,\phi_i(b_{i,\leq t}),\,\theta_{-i,\leq t}\Big]\\
        &\leq\sum_{t\in\T_i}\gamma_i(t)\expec_{s_t}\!\Big[u_i(\theta_{i,t},s_t)\,\Big|\,\phi_i\circ\phi_i(b_{i,\leq t}),\,\theta_{-i,\leq t}\Big]\\
        &\leq\cdots
    \end{align*}
    Since there are only finitely many rounds, after a finite number of applications of $\phi_i$, the above process reaches to a fixed point $b_{i,\leq T}=\phi_i(b_{i,\leq T})$ such that $b_{i,t}=\theta_{i,t}$ for all $t\in\T_i$. Hence,
    \begin{align*}
        u_i\!\rbra{\sigma_i,\theta_{-i};\Mwdp,A}&\leq\sum_{t\in\T_i}\gamma_i(t)\expec_{s_t}\!\Big[u_i(\theta_{i,t},s_t)\,\Big|\,\hat{b}_{i,\leq t},\,\theta_{-i,\leq t}\Big],\tag{by \eqref{eq:u_i-holder}}\\
        &\leq\sum_{t\in\T_i}\gamma_i(t)\expec_{s_t}\!\Big[u_i(\theta_{i,t},s_t)\,\Big|\,\theta_{i,\leq t},\,\theta_{-i,\leq t}\Big]\\
        &=u_i\!\rbra{\theta_i,\theta_{-i};\Mwdp,A}.
    \end{align*}
    Therefore, the strategy profile $\theta:=(\theta_i)_{i\in[n]}$ is a Nash equilibrium.
\end{proof}

\begin{lemma}[Post-processing \cite{dwork2014algorithmic}]
    Let $f:Y\rightarrow Y'$ be an arbitrary function. If $q$ is $\eta$-DP then $f\circ q:X^m\rightarrow\Delta(Y')$ is also $\eta$-DP.
    \label{lem:post-dp}
\end{lemma}
\section{Proof of Theorem \ref{thm:hedge-mech}}\label{app:proof:thm:hedge-mech}
%%%%%%%%%%%%%%%%%%%%%%%%%%%%%%%%%%%%%%%%%%%
In this appendix, we give a proof of Theorem \ref{thm:hedge-mech} restated below.

\HedgeLearner*
\begin{proof}
    Recall that $\Mhedge$ is given as follows: For given $T\in\nat$, $\Mhedge$ outputs the following mechanism for each $t\in[T]$:
    \begin{align*}
        \pi_t(\cdot):=(1-\lambda)\cdot\pisr_t(\cdot)+\lambda\cdot\picom(\cdot),
    \end{align*}
    where $\pisr_t\sim\qhedge_{\eta,t}(b_{<t},G_{<t})$ and $\big(\qhedge_{\eta,t}(b_{<t},G_{<t})\big)_{t\in[T]}$ is the exponential weights sequence~\eqref{def:hedge-update} with objectives $F_t(\theta_t,\pi):=\expec_{s\sim\pi(\theta_t)}G_t(\theta_t,s)$.
    
    First, we show that, for every $T\in\nat$, there is a choice of $(\lambda,\eta)$ such that the truth-telling profile is a Nash equilibrium of any game induced by $\Mhedge$ and $A\in\A_T$. By Lemma \ref{lem:hedgedp} and Theorem \ref{thm:ic-via-dp}, this holds whenever
    \begin{align*}
        \frac{\lambda \beta(\picom)}{16\eta}\geq\sup_{A\in\A_T}\alpha(A):=\alpha_T(\A)\ \text{ and }\ 4\eta\in\unit.
        \numberthis\label{eq:nic-cond}
    \end{align*}
    The choice $\eta:=\min\cbra{1/\sqrt{\alpha_T(\A)\,T},\ \beta(\picom)/(16\alpha_T(\A))}$ and $\lambda:=16\eta\,\alpha_T(\A)/\beta(\picom)$ satisfies \eqref{eq:nic-cond} for every $T\in\nat$: $\lambda\leq1$ because $\eta\leq\beta(\picom)/(16\alpha_T(\A))$, and $4\eta\leq\beta(\picom)/(4\alpha_T(\A))\leq1/2$ because $\beta(\picom)\leq2$ for utilities in $\munit$ and $\alpha_T(\A)\geq1$. Moreover, to choose such $\lambda$s, one only needs the knowledge of the whole class $\A$; hence such choices can be made before starting any round.

    Next, we show that $\Mhedge$ guarantees no-regret: $\sup_{A\in\A_T}\reg\big(\theta_{\leq T},\Mhedge,\Pi;A\big)\in o(T)$. By the linearity of expectation and the standard regret bound of the exponential weights algorithm \cite{arora2012multiplicative,hazan2016introduction}, we have that
    \begin{align*}
        \sup_{A\in\A_T}\reg\big(\theta_{\leq T},\Mhedge,\Pi;A\big)&\leq 4\eta\,T+\frac{1}{\eta}\log|\Pi|+2\lambda\,T.
    \end{align*}
    If $\eta=1/\sqrt{\alpha_T(\A)\,T}$, then $\lambda=16\sqrt{\alpha_T(\A)/T}/\beta(\picom)$, and the right-hand side of the above bound is in $O\big((\log|\Pi|+\beta(\picom)^{-1})\cdot\sqrt{\alpha_T(\A)\,T}\big)$. Otherwise, $T<256\,\alpha_T(\A)/\beta(\picom)^2$, and since the objectives take values in $\munit$, the regret is at most $2T<32\,\beta(\picom)^{-1}\sqrt{\alpha_T(\A)\,T}$.
\end{proof}
\section{Application: Online Optimal Resource Allocation}\label{app:resource-alloc}
In this appendix, we provide an additional application of Theorem \ref{thm:hedge-mech}: Designing a Nash incentive-compatible online mechanism for optimal resource allocation. 

\paragraph{Setting.}
Consider there are $k$ CPUs in a cluster shared by $n$ users. User $i$'s type $\theta_i\in[k]$ represents her true demand on CPU. The social outcome $s$ is the CPU allocation vector, that is, $s_i\in\nat$ represents the number of CPUs allocated to user $i$. Hence, the outcome space is $S_{n,k}:=\big\{s\in\nat^n:\sum_{i\in[n]}s_i\leq k\big\}$. We assume all users have identical utility functions given by $u_i(\theta_i,s):=\ind\sbra{s_i\geq \theta_i}$, that is, an agent experiences utility $1$ with enough resources and $0$ for insufficient resources. Then, the cluster manager's objective function is the fraction of jobs completed, weighted by job importance. That is, $G_t(\theta_t,s):=\frac{1}{n}\sum_{i\in[n]}r_{i,t}\cdot u_i(\theta_{i,t},s)$, where $r_{i,t}\in\unit$s are \textit{ex-post} weighting factors (such as \textit{ex-post} resource utilization rates), specified by the cluster manager.

\paragraph{Target mechanisms.} The cluster manager chooses any DSIC/NIC off-the-shelf single-round resource allocation mechanism with parameters. An example would be the max-min fair allocation mechanism $\Pimaxmin_{n,k}$ whose parameter is the initial endowment vector $w\in S_{n,k}$~\citep{demers1989analysis,kandasamy2020online}. Another example would be the posted-allocation mechanism whose outcome is always its parameter vector $w\in S_{n,k}$. With any choice of parametrized mechanism, the cluster manager's goal is to compete against the best parameter in terms of the cumulative social welfare $\sum_{t\in[T]}G_t(\theta_t,s)$.

\paragraph{Commitment mechanism.} Our commitment mechanism $\picom$ is as follows: Given input demands $b\in[k]^n$, $\picom(b)$ samples $i'\in [n]$ and $r\in [2k]$ uniformly at random. If $r> b_{i'}$, set $s_{i'}=b_{i'}$ and $s_j=\big\lfloor\frac{k-b_{i'}}{n-1}\big\rfloor$ for all $j\neq i'$ with $b_j\neq\perp$. Otherwise, set $s_{i'}=0$ and $s_j=\big\lfloor\frac{k}{n-1}\big\rfloor$ for all $j\neq i'$ with $b_j\neq\perp$. Then, we show the following.
\begin{lemma}
    For the commitment mechanism described above, we have that $\beta(\picom)\geq (2nk)^{-1}$.
\end{lemma}
\begin{proof}
    First, recall the definition of $\beta(\picom)$ here. The penalty gap $\beta(\picom)$ is the minimum of the following quantity,
    \begin{align*}
        \expec_{s\sim\pi(\theta_i,b_{-i})}u_i(\theta_i,s)-\expec_{s\sim\pi(b_i,b_{-i})}u_i(\theta_i,s),
    \end{align*}
    over all $i$, $b_{-i}$ and $\theta_i,b_i\in\Theta_i$ such that $b_i\neq\theta_i$. Fix an $i\in[n]$ and $b_{-i}$. Then, consider two cases $b_i>\theta_i$ and $b_i<\theta_i$.
    
    \textbf{Case $b_i>\theta_i$}.\ \ In this case, we have that
    \begin{align*}
        \expec_{s\sim\pi(b_i,b_{-i})}u_i(\theta_i,s)&=\rbra{1-\frac{1}{n}}\cdot C + \frac{1}{n}\cdot\frac{2k-b_i}{2k}\\
        &=\rbra{1-\frac{1}{n}}\cdot C +\frac{1}{n}\rbra{\frac{2k-\theta_i}{2k}-\frac{b_i-\theta_i}{2k}}\\
        &=\bunderbrace{\rbra{1-\frac{1}{n}}\cdot C +\frac{1}{n}\cdot\frac{2k-\theta_i}{2k}}{\expec_{s\sim\pi(\theta_i,b_{-i})}u_i(\theta_i,s)}-\frac{b_i-\theta_i}{2nk}\\
        &\leq\expec_{s\sim\pi(\theta_i,b_{-i})}u_i(\theta_i,s)-\frac{1}{2nk}\tag{as $b_i-\theta_i\geq 1$}
    \end{align*}
    where $C$ is a constant that does not depend on the input demand of user $i$. Hence, we have that
    \begin{align*}
        \expec_{s\sim\pi(\theta_i,b_{-i})}u_i(\theta_i,s)-\expec_{s\sim\pi(b_i,b_{-i})}u_i(\theta_i,s)\geq \frac{1}{2nk},
    \end{align*}
    when $b_i>\theta_i$.

    \textbf{Case $b_i<\theta_i$}.\ \ In this case, observe that
    \begin{align*}
        \expec_{s\sim\pi(b_i,b_{-i})}u_i(\theta_i,s)&=\rbra{1-\frac{1}{n}}\cdot C\\
        &=\rbra{1-\frac{1}{n}}\cdot C + \frac{1}{n}\cdot\frac{2k-\theta_i}{2k}-\frac{1}{n}\cdot\frac{2k-\theta_i}{2k}\\
        &=\expec_{s\sim\pi(\theta_i,b_{-i})}u_i(\theta_i,s)-\frac{1}{n}\cdot\frac{2k-\theta_i}{2k}\\
        &\leq \expec_{s\sim\pi(\theta_i,b_{-i})}u_i(\theta_i,s)-\frac{1}{2n}\tag{as $\theta_i\leq k$}
    \end{align*}
    Hence, we have that 
    \begin{align*}
        \expec_{s\sim\pi(\theta_i,b_{-i})}u_i(\theta_i,s)-\expec_{s\sim\pi(b_i,b_{-i})}u_i(\theta_i,s)\geq \frac{1}{2n}\geq\frac{1}{2nk},
    \end{align*}
    when $b_i<\theta_i$.
\end{proof}

\paragraph{Regret bound.} To sum up, we have that $\beta(\picom)\geq (2nk)^{-1}$ for this problem. If the cluster manager's goal is to maximize the social welfare with respect to $\Pimaxmin_{n,k}$, our Theorem \ref{thm:hedge-mech} implies there is a \NICOM\ for $\Pimaxmin_{n,k}$ with regret $O\Big( (n\log k + nk)\cdot T^{\frac{1+h}{2}} \Big)$ under any agent class $\A$ such that $\alpha_T(\A)\in O(T^h)$ for some $h\in\unito$.
\section{Related works}\label{app:related-work}
\paragraph{Dynamic pricing.}
Online learning for auctions traces back to the seminal work of \citet{kleinberg2003value}. Online learning of posted-price auctions is often referred to as dynamic pricing, and there is a wealth of literature on this subject following the aforementioned seminal works. Therefore, we direct readers to a review of this topic~\citep{den2015dynamic}.

\paragraph{Strategy-robust pricing.}
One of the first works to investigate non-myopic agents was by \citet{amin2013learning}, who introduced a notion of regret that remains robust under strategic buyers. They devised a robust pricing algorithm for the repeated single-bidder single-item posted price mechanism, assuming that the buyer's value is drawn independently from an unknown value distribution.

This setting has subsequently been generalized in two ways. The first is the contextual setting, where the seller observes contextual information (such as UserID) while selling an item. \citet{amin2014repeated} were the first to study the problem of robustly learning prices under the contextual setting. \citet{drutsa2020optimal}, \citet{zhiyanov2020bisection}, \citet{golrezaei2019dynamic}, and \citet{deng2019robust} further optimized regret bounds or extended the problem setting by considering noisy observations.

While previous works relied on distributional assumptions, a generalization of the previous stochastic setting to an online adversarial setup was explored by \citet{liu2018learning}. Particularly, \citet{liu2018learning} extended this to a multi-bidder second-price auction in an adversarial environment, using a strongly differentially private algorithm for the expert problem. In comparison to our work, previous studies focus on scenarios where the objective is to learn a single parameter (price or reserve price).

Motivated by online advertising, a line of work has studied bandit learning in single-item auctions under truthfulness constraints~\citep{babaioff2009characterizing,babaioff2015truthful,devanur2009price}.
Here, buyers place a single bid at the beginning of the game representing their value for the bid, and on each round, the seller observes the number of clicks for the buyer she chooses for the round.
When compared to our setting, as the buyers place only a single bid at the beginning, they have significantly less opportunity to manipulate outcomes in their favor.

\paragraph{Differential privacy and mechanism design.}
Since the seminal work of \citet{mcsherry2007mechanism}, differential privacy and mechanism design have become closely intertwined fields. Subsequently, \citet{nissim2012approximately} and \citet{kearns2014mechanism} expanded the scope of mechanism design through differential privacy. Notably, \citet{nissim2012approximately} introduced the commitment mechanism and constructed general mechanisms from the exponential mechanism, including mechanisms without money~\citep{moulin1980strategy,procaccia2013approximate}. Our work can be seen as a conceptual extension of \citet{nissim2012approximately}'s work to an online adversarial setting. As discussed earlier, \citet{liu2018learning} employed differential privacy to design a strategy-robust online learning algorithm for single-item auctions. \citet{abernethy2019learning} also applied a differentially private primitive to design a strategy-robust version of Myerson's optimal auction in the online setting (see~\S\ref{subsec:relevant-work}). For further exploration of the growing trend of using differential privacy to control incentives, we recommend reading a review by \citet{zhang2021more}.

\section{Proof of Proposition \ref{prop:lower-bound}}\label{app:proof:prop:lower-bound}
In this appendix, we give a proof of Proposition~\ref{prop:lower-bound} restated below. The proof follows the proof of Theorem~3 of \citet{amin2013learning}. That theorem bounds the strategic regret of a seller who posts prices to a buyer whose values are drawn from a fixed distribution. As in their proof, weighting the rounds by the discount function turns the online mechanism into a single-round mechanism that is incentive compatible and individually rational, and for at least one type, the revenue of such a mechanism is smaller than the type by a constant. Our regret is measured under truthful reports, so the agent never buys at a price above her type, and we do not need Lemma~4 of \citet{amin2013learning}. Moreover, in the instances below, the agent has type $0$ before a round $t^\star$, so the bound depends on $\alpha_T(\A)$ rather than on the sum $\sum_{t\in[T]}\gamma_1(t)$ that appears in the bound of \citet{amin2013learning}.

\PropLowerBoundNIC*

\begin{proof}
    Write $\gamma:=\gamma_1$ and $H_t:=\frac{1}{\gamma(t)}\sum^T_{\tau=t}\gamma(\tau)$ for $t\in[T]$. Every $A\in\A$ has $\T_1=[T]$, and $H_t\leq T-t+1\leq T$ because $\gamma$ is non-increasing, so $\alpha(A)=\max_{t\in[T]}H_t$ by Definition~\ref{def:long-sightedness}. Fix $t^\star\in\argmax_{t\in[T]}H_t$, so that $H_{t^\star}=\alpha_T(\A)$, and let $\Gamma:=\sum^T_{t=t^\star}\gamma(t)=\gamma(t^\star)H_{t^\star}$.

    For $v\in\{\frac{1}{2},1\}$, let $A^v\in\A$ have types $\theta^v_{1,t}:=0$ for $t<t^\star$ and $\theta^v_{1,t}:=v$ for $t\geq t^\star$. When $A^v$ reports truthfully, let $p^v_t$ be the price posted on round $t$ (for a lottery over posted-price mechanisms, the price that the lottery draws), and let $x^v_t:=\ind\!\sbra{\theta^v_{1,t}\geq p^v_t}$ indicate a sale. The probability of a sale and the revenue from round $t^\star$ on, weighted by the discount function, are
    \begin{align*}
        x(v):=\sum^T_{t=t^\star}\frac{\gamma(t)}{\Gamma}\,\expec\sbra{x^v_t}\in[0,1]
        \qquad\text{and}\qquad
        R(v):=\sum^T_{t=t^\star}\frac{\gamma(t)}{\Gamma}\,\expec\sbra{p^v_t x^v_t}.
    \end{align*}
    A sale on a round $t\geq t^\star$ requires $p^v_t\leq v$, so $p^v_t x^v_t\leq v\,x^v_t\leq v$. Hence
    \begin{align}
        \expec\sbra{p^v_t x^v_t}\leq v\ \text{ for } t\geq t^\star,
        \qquad\text{and}\qquad
        R(v)\leq v\,x(v).
        \label{eq:lb-ir}
    \end{align}

    \textit{Regret.} On rounds $t<t^\star$, the type is $0$, so a sale happens only at price $0$, and the revenue is $0$ for $\M$ and for every posted-price mechanism. On rounds $t\geq t^\star$, the mechanism $\pipp_v\in\Pi$ sells at price $v$, so $\opt(\Pi;A^v)\geq v\,(T-t^\star+1)$. Hence
    \begin{align*}
        \reg\big(\theta^v_{\leq T},\M,\Pi;A^v\big)
        \geq\sum^T_{t=t^\star}\Big(v-\expec\sbra{p^v_t x^v_t}\Big)
        \geq\sum^T_{t=t^\star}\frac{\gamma(t)}{\gamma(t^\star)}\Big(v-\expec\sbra{p^v_t x^v_t}\Big)
        =H_{t^\star}\big(v-R(v)\big),
        \numberthis\label{eq:lb-reg}
    \end{align*}
    where the second inequality holds because each term is nonnegative by~\eqref{eq:lb-ir} and $\gamma(t)\leq\gamma(t^\star)$ for $t\geq t^\star$, and the equality uses $\Gamma=\gamma(t^\star)H_{t^\star}$.

    \textit{Incentive compatibility.} Let $\sigma'_1$ be the strategy of $A^1$ that reports $\theta^{1/2}_{1,t}$ on every round $t$. Since the online mechanism observes only the reports and the objectives, and a lottery draws its price independently of the report, the prices and sales under $\sigma'_1$ have the same joint distribution as $(p^{1/2}_t,x^{1/2}_t)_{t\in[T]}$, the prices and sales under truthful reporting by $A^{1/2}$. On rounds $t<t^\star$, the type of $A^1$ is $0$ and both strategies report $0$, so both give her utility $0$. On rounds $t\geq t^\star$, her type is $1$, so when she reports truthfully she buys at every price in $[0,1]$, and $x^1_t=1$. Therefore
    \begin{align*}
        u_1(\theta_1;\M,A^1)&=\Gamma\big(1-R(1)\big),\\
        u_1(\sigma'_1;\M,A^1)&=\sum^T_{t=t^\star}\gamma(t)\,\expec\sbra{x^{1/2}_t\big(1-p^{1/2}_t\big)}=\Gamma\big(x(\tfrac{1}{2})-R(\tfrac{1}{2})\big).
    \end{align*}
    Since $\M$ is \NIC\ with respect to $A^1$, the first is at least the second. Together with $x(\frac{1}{2})\geq 2R(\frac{1}{2})$ from~\eqref{eq:lb-ir}, this gives $1-R(1)\geq x(\frac{1}{2})-R(\frac{1}{2})\geq R(\frac{1}{2})$, that is,
    \begin{align*}
        \big(1-R(1)\big)+\big(\tfrac{1}{2}-R(\tfrac{1}{2})\big)\geq\tfrac{1}{2}.
    \end{align*}
    Hence $v-R(v)\geq\frac{1}{4}$ for some $v\in\{\frac{1}{2},1\}$, and for this $v$, \eqref{eq:lb-reg} gives
    \begin{align*}
        \sup_{A\in\A}\reg\big(\theta_{\leq T},\M,\Pi;A\big)\geq\reg\big(\theta^v_{\leq T},\M,\Pi;A^v\big)\geq\frac{H_{t^\star}}{4}=\frac{\alpha_T(\A)}{4}.
    \end{align*}
\end{proof}

%%%%%%%%%%%%%%%%%%%%%%%%%%%%%%%%%%%%%%%%%%%%%%%%%%%%%%%%%%%%%%%%%%%%%%%%%%%%%%%
%%%%%%%%%%%%%%%%%%%%%%%%%%%%%%%%%%%%%%%%%%%%%%%%%%%%%%%%%%%%%%%%%%%%%%%%%%%%%%%
\end{document}